\documentclass{article}
\input{preamble.sty}

\title{Prior-Independent Bidding Strategies for First-Price Auctions}
\author{
\sf Rachitesh Kumar\\
\sf Columbia University\\
\texttt{rk3068@columbia.edu}
\and
\sf Omar Mouchtaki\\
\sf New York University\\
\texttt{om2166@stern.nyu.edu }
}
\date{\today}

\begin{document}

\maketitle

\begin{abstract}
    First-price auctions are one of the most popular mechanisms for selling goods and services, with applications ranging from display advertising to timber sales. Unlike their close cousin, the second-price auction, first-price auctions do not admit a dominant strategy. Instead, each buyer must design a bidding strategy that maps values to bids---a task that is often challenging due to the lack of prior knowledge about competing bids. To address this challenge, we conduct a principled analysis of prior-independent bidding strategies for first-price auctions using worst-case regret as the performance measure. First, we develop a technique to evaluate the worst-case regret for (almost) any given value distribution and bidding strategy, reducing the complex task of ascertaining the worst-case competing-bid distribution  to a simple line search. Next, building on our evaluation technique, we minimize worst-case regret and characterize a minimax-optimal bidding strategy for \emph{every value distribution}. We achieve it by explicitly constructing a bidding strategy as a solution to an ordinary differential equation, and by proving its optimality for the intricate infinite-dimensional minimax problem underlying worst-case regret minimization. Our construction provides a systematic and computationally-tractable procedure for deriving minimax-optimal bidding strategies. When the value distribution is continuous, it yields a deterministic strategy that maps each value to a single bid. We also show that our minimax strategy significantly outperforms the uniform-bid-shading strategies advanced by prior work.
    Importantly, our result allows us to precisely quantify, through minimax regret, the performance loss due to a lack of knowledge about competing bids. We leverage this to analyze the impact of the value distribution on the performance loss, and find that it decreases as the buyer's values become more dispersed.
\end{abstract}

\section{Introduction}\label{sec:intro}
First-price auctions are a dominant mechanism across various markets, playing a central role in how goods and services are allocated. This time-tested mechanism has gained renewed interest following its recent adoption by the display advertising industry, replacing second-price auction. Unlike second-price auctions, where truthfully bidding one's value is a dominant strategy, i.e., it is optimal regardless of the competitor’s actions, first-price auctions require bidders to engage in far more complex strategic considerations. Thus, this change has brought to the forefront fundamental questions about how buyers should navigate these strategic intricacies to maximize their utility. Our work takes the perspective of such a buyer, and develops a principled approach to answering these questions.

The primary challenge a buyer faces in first-price auctions is that the outcome of her bid depends heavily on the competing bids she will encounter. Thus, a buyer must anticipate how others will bid in order to determine her own bid, often without knowing the competitors' values or strategies. Traditional game-theoretic approaches to the design of bidding strategies rely on strong assumptions like common knowledge and equilibrium behavior, which fail to hold in real-life settings~\citep{kasberger2023robust}. This leaves open the practical problem of coming up with good strategies despite only having limited knowledge about the competition. To address this problem, we take a prior-independent approach to the design of bidding strategies for first-price auctions. Rather than making assumptions about the behavior of competing buyers, we seek to develop a principled methodology for evaluating, comparing, and designing bidding strategies that perform well across all possibilities. Our approach is inspired by distributional robust paradigms, where the objective is to ensure strong bidding performance under minimal distributional assumptions. By doing so, we aim to provide a framework that is both theoretically sound and practically relevant for buyers who must make bidding decisions in complex and uncertain market conditions.

\subsection{Main Contributions}

We consider a buyer participating in a first-price auction who wishes to design a bidding strategy---a mapping from her private value to a (potentially random) bid---that maximizes her expected utility. 
We adopt a prior-independent approach which does not make any assumptions on the competing bids, and instead aims to optimize performance uniformly against all possible competing bids. To jointly evaluate the performance across all possibilities, we use the standard metric of worst-case regret \citep{savage1951theory}. In particular, for any distribution of the highest competing bid, we define regret to be the difference between the expected utility achievable by the oracle who knows this distribution and the one generated by the strategy of interest. We focus on the highest competing bid as it is a sufficient statistic that completely determines the utility generated by any strategy. The worst-case regret of a strategy then is the maximum regret it incurs across all possible highest-competing-bid distributions. It captures the loss incurred by the strategy due to a lack of knowledge about competing bids, which is a core concern in first-price auctions. A small worst-case regret implies that the strategy does not incur a large loss due this lack of information no matter how the competing bids are determined, thereby circumventing the challenging task of accurately understanding and predicting the behavior of competitors.

\textbf{Performance evaluation.} Before minimizing worst-case regret, one must tackle the task of evaluating it for a given bidding strategy. The infinite-dimensional nature of bidding strategies and highest-competing-bid distributions makes this task challenging. These difficulties are further exacerbated by the fact that our benchmark---utility achievable by the oracle who knows the highest-competing-bid distributions---is itself the value of an infinite-dimensional optimization problem, making even the task of characterizing regret for a fixed highest-competing-bid distribution difficult. In \Cref{thm:evaluation}, we show that the problem of evaluating worst-case regret can be considerably simplified for a wide class of strategies and value distributions: the dimension of the underlying optimization problem can be reduced from infinity to one. Firstly, from a computational perspective, it reduces an a priori intractable problem to a remarkably simple one that can be solved with a line search. Secondly, this result has an insightful economic interpretation: we show that the worst-case highest-competing-bid distribution is always a deterministic one. In other words, when designing strategies with small worst-case regret, one can focus on deterministic highest competing bids and ignore their potential for random variation. We leverage \Cref{thm:evaluation} to compare different uniform-bid-shading strategies, an important class of practical bidding strategies that bid $\alpha \cdot v$ when the buyer's value is $v$, for some fixed $\alpha \in [0,1]$. In particular, we show that the uniform-bid-shading strategy with $\alpha = 0.5$, which emerges in \citet{kasberger2023robust} when the buyer robustly optimizes the worst-case regret for each value in isolation, can be considerably improved upon by accounting for the distribution of values and choosing the optimal $\alpha$ for it. For instance, when the value distribution is uniform, we show that the best choice of $\alpha$ is $0.38$, and this choice yields a worst-case regret that is $20\%$ lower than the common choice of $0.5$.

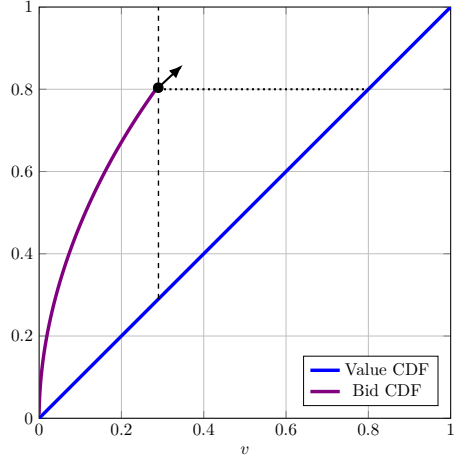
\begin{figure}[h!]
    \centering
    \begin{tikzpicture}[scale = 0.65]
            \begin{axis}[
        width=10cm,
        height=10cm,
        xmin=0,xmax=1,
        ymin=0,ymax=1,
        table/col sep=comma,
        xlabel={$v$},
        ylabel={},
        grid=both,
        legend pos=south east,
    ]

    \addplot [blue, line width = 0.7mm] {x}; 
    \addlegendentry{Value CDF}

    \addplot [violet,  line width=0.7mm, select coords between index={1}{79}] table[x={Qs_alpha=1.0},y={x_alpha=1.0}] {Data/plot_bidding_strategy_beta.csv};
    \addlegendentry{Bid CDF}
    
    \addplot [Circle-Latex, very thick,  black] coordinates { (0.28,0.795) (0.35,0.86)};

    \draw[dashed, thick, black] (axis cs:0.29,0.29) -- (axis cs:0.29,1);
    \draw[dotted, very thick, black] (axis cs:0.29,0.8) -- (axis cs:0.8,0.8);

    \end{axis}
    \end{tikzpicture}
    \caption{\textbf{ODE which constructs a minimax-optimal bidding strategy for the uniform value distribution.} If one starts at (0,0) and moves with a slope equal to the ratio of the dashed line to dotten line, then the resulting curve will trace out the CDF of bids under a minimax-optimal bidding strategy. The strategy is then to simply bid the corresponding quantile for every value, i.e., bid $b$ for value $v$ if and only if the quantiles of $b$ and $v$ are equal under the bid and value CDFs respectively.}
    \label{fig:intro_ode}
\end{figure}

\textbf{Minimax-optimal bidding strategy.} Having developed an efficient method for evaluating worst-case regret, we next turn towards optimizing it. We provide a complete characterization of minimax-optimal regret, which is the smallest-possible worst-case regret achievable by any bidding strategy, and do so for \emph{every} value distribution. Our characterization takes the form of an explicit construction of a saddle point for the underlying minimax optimization problem using ordinary differential equations (ODEs). On the technical front, this requires multiple advances. First, we leverage our result on performance evaluation to simplify the benchmark from the optimal utility achievable with knowledge of the highest-competing-bid distribution to that achievable with knowledge of its realization. Moreover, we simplify the space of bidding strategies via a reformulation that assigns a deterministic bid to each quantile of the value distribution instead of a random bid to each value. Then, we consider the first-order optimality conditions of this reformulation, and analyze the resulting ordinary differential equations. Our primary technical contributions pertain to the analysis of these ODEs and addressing the associated challenges: (i)~the worst-case-regret minimization problem is parameterized by the value distribution of the buyer and so are the ODEs; (ii)~discontinuities in the value distribution manifest as discontinuities in the ODEs; (iii)~the ODEs have ill-conditioned denominators prone to divergence and they are not even well-defined everywhere. To establish our main result (\Cref{thm:main-result}), we navigate these hurdles to characterize the saddle point as a solution to these ODEs. It yields an efficient procedure for constructing minimax-optimal bidding strategies for arbitrary value distributions; \Cref{fig:intro_ode} illustrates it for the uniform value distribution. Once the minimax-optimal strategy has been constructed, we show that the corresponding value of optimal regret is given by a simple integral. Altogether, our characterization provides an efficient technique for solving the intricate infinite-dimensional minimax optimization problem that arises in the prior-independent setting.

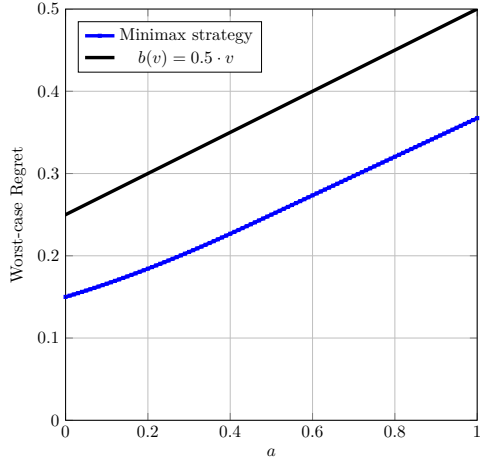
\begin{figure}[h!]
    \centering
\begin{tikzpicture}[scale = 0.65]
    \begin{axis}[
        width=10cm,
        height=10cm,
        xmin=0,xmax=1,
        ymin=0,ymax=0.5,
        table/col sep=comma,
        xlabel={$a$},
        ylabel={Worst-case Regret},
        grid=both,
        legend pos=north west,
    ]

    \addplot [blue,  line width=0.7mm, ,mark=square,mark options={scale=.1}] table[x=a,y={regret}] {Data/regrets_uniform_a.csv};
    \addlegendentry{Minimax strategy}
    \addplot [black, line width = 0.7mm] {0.25*x + 0.25};  
   \addlegendentry{$b(v) = 0.5 \cdot v$}
      \addlegendimage{ultra thick,black}
    \end{axis}
    \end{tikzpicture}
    \caption{\textbf{Summary of our insights.} For every $a$, the blue curve represents the largest worst-case regret incurred by our minimax-optimal bidding strategy across all value distributions with a density bounded above $\frac{1}{1-a}$. The black curve corresponds to the worst-case regret of the bidding strategy from \citet{kasberger2023robust} which bids $0.5 \cdot v$ for every value $v$, when the value distribution is a uniform on $[a,1]$.}
    \label{fig:intro_impact_value}
\end{figure}

\textbf{Structural Insights.} The procedure we develop for constructing minimax-optimal bidding strategies delivers benefits beyond computational tractability. It allows us to glean structural insights about the drivers of performance loss, as measured by regret. First, when the value distribution is continuous, our construction yields a \emph{deterministic} minimax-optimal bidding strategy. In contrast, deterministic strategies are sub-optimal when the value is deterministic. In this case, the problem is equivalent to robust pricing \citep{bergemann2011robust}, inheriting the associated sub-optimality of deterministic strategies. Since any distribution can be perturbed ever so slightly to arrive at a continuous one, our result implies the optimality of deterministic strategies for a large class of value distributions. In particular, even though deterministic strategies are sub-optimal when the buyer's value is deterministic and known with certainty, they immediately jump to optimality in the presence of even an infinitesimal amount of (continuous) random noise in the value estimate. Furthermore, we characterize the impact of the value distribution on minimax regret. Specifically, we show in \Cref{thm:worst-value-dist} that the uniform distribution on $[1 - \tfrac{1}{\rho}, 1]$ yields the highest minimax-optimal regret among all value distributions with a density bounded above by $\rho \geq 1$. This allows us to evaluate minimax regret as a function of how concentrated the values are: we find that a greater dispersion in values leads to lower minimax-optimal regret; see \Cref{fig:intro_impact_value}. Taken together, these insights indicate that even a small amount of variation in values is often sufficient to render deterministic strategies optimal, and the performance of the optimal policy improves with the amount of variation. Hence, the presence of random private information in the form of values, which is known to the buyer but not the competition, obviates the need to hedge bids with randomization. Moreover, the optimal achievable performance improves with the magnitude of randomness.

\subsection{Related Work}

Starting from the seminal work of \citet{vickrey1961counterspeculation}, first-price auctions have received significant attention in the literature. Most of this work has focused on the equilibrium analysis of multi-buyer interactions. In contrast, our focus is on developing bidding strategies for an individual buyer which are robust to the behavior of the competition, without regard for how the competition arrives at that behavior. Therefore, we do not review the vast literature on the traditional equilibrium analysis and refer the reader to standard texts~\citep{krishna2009auction, milgrom2004putting}.

The closest works to ours are the recent ones of \citet{kasberger2023robust} and \citet{qu2024double}. \citet{kasberger2023robust} study robust bidding in first-price auctions, with the goal of providing practical guidance for real-life auctions. They provide comprehensive evidence on the need to go beyond traditional analyses in the form of surveys, laboratory data and empirical analysis. On the technical front, they consider a buyer with a fixed value who is uncertain about the bids of others, and construct deterministic bids which achieve low worst-case regret with respect to the uncertainty in competing bids. Beyond modeling uncertainty directly in the highest competing bid, they also consider models with higher order information about how the competing bids are generated. We focus on uncertainty in the highest competing bids and extend their results along two dimensions: (i) we allow for randomized bidding strategies; (ii) we allow the value to be random, measure regret in expectation over this randomness, and characterize minimax-optimal strategy for every value distribution. Both extensions create significant challenges by replacing finite-dimensional optimizations with infinite-dimensional ones. \citet{qu2024double} adopts a distributionally-robust approach. They optimize a single bid to maximize worst-case expected utility over value and highest-competing-bid distributions lying within a Kullback-Leibler ball around empirical estimates. In contrast, we assume that buyer knows her own value and can alter the bid based on the value, which results in an action space consisting of bidding strategies instead of a single bid. Moreover, we do not restrict the highest-competing-bid to lie some known neighborhood, and we use regret as our metric, which, unlike absolute utility, is not linear in the highest-competing-bid distribution.

Another line of work develops learning algorithms for a buyer participating in repeated first-price auctions, either in stochastic settings \citep{han2020optimal,balseiro2022contextual,badanidiyuru2023learning,schneider2024optimal} or adversarial ones \citep{han2020learning,zhang2022leveraging,kumar2024strategically}. In these works, the benchmark for regret is the optimal fixed strategy in hindsight. Although we do not explicitly model repeated auctions or learning dynamics, our results imply a regret guarantee against the optimal \emph{sequence of strategies} in hindsight without making any assumptions on the environment. Our minimax-strategy serves as a natural choice for settings with a high degree of uncertainty and churn that make learning impossible. It can also be used to warm-start learning algorithms to improve performance early on. Finally, our work contributes to the literature on decision-making under uncertainty via distributionally-robust regret. This approach provides structural insights into robust decision-making when faced with limited information, and has been applied in various contexts, such as robust pricing \citep{bergemann2011robust}, inventory management \citep{perakis2008regret}, auction design \citep{anunrojwong2022robustness, anunrojwong2023robust}, and bidding \citep{kasberger2023robust}.

\section{Model}\label{sec:model}

\paragraph{Notation.} We associate the Borel sigma algebra with the set $[0,1]$, and denote it by $\borel$. We use $\lambda(\cdot)$ to represent the Lebesgue measure.  For any set $A$, $\Delta \left( A \right)$ denotes the set of probability measures on $A$. For every $a \in A$, we denote by $\delta_a$ the Dirac distribution  which puts all mass at $a$. Given a space of probability measures $\Delta(A)$, we refer to the topology induced by the weak convergence on that space as the ``weak topology''. Unless otherwise specified, product spaces are endowed with the product sigma algebra and the product topology. For a cumulative distribution function (CDF) $F:[0,1] \to [0,1]$, we use $F^-$ to denote its generalized inverse, i.e., $F^-(t) \coloneqq \inf\{x \mid F(x) \geq t\}$. For $A,B \subset \mathbb{R}$, we say that $A \preccurlyeq B$ if and only if $\sup A \leq \inf B$. We use $\{f(X) \mid X \sim \mathcal D\}$ to denote the distribution of $f(X)$ when $X \sim \mathcal D$.

Consider a buyer participating in a sealed-bid first-price auction for a single indivisible good. Let $v$ denote the value the buyer derives from winning the good. We assume that $v \in [0,1]$ and it is distributed according to the distribution $F \in \Delta([0,1])$; we use $F$ to denote both the CDF and the measure it defines. Given a value ${v \sim F}$, the buyer selects a (potentially random) bid $b \sim s(v)$, where ${s: [0,1] \longrightarrow \Delta([0,1])}$ is the buyer's bidding strategy that maps value $v$ to the bid distribution $s(v)$. For the sake of completeness and rigor, we only consider bidding strategies for which there exists a Markov kernel\footnote{$\kappa: \borel \times [0,1] \to [0,1]$ is called a \emph{Markov kernel} if (i) for every fixed $B \in \borel$, $v \mapsto \kappa(B,v)$ is measurable; (ii) for every fixed $v \in [0,1]$, $\kappa(\cdot, v)$ is a probability measure on $([0,1], \borel)$.} ${\kappa_s:\borel \times [0,1] \to [0,1]}$ such that $\Prob_{b \sim s(v)}(B) \coloneqq \kappa_s(B, v)$ for all $B \in \borel$; let $\Scal$ denote the set of all such bidding strategies.
We use $J_{s,F}$ to denote the joint distribution of value-bid pairs $(v,b)$ under the strategy $s$, i.e., ${J_{s,F}(A \times B) = \E_{v \sim F}[\kappa_s(B,v) \cdot \mathbbm{1}(v \in A)]}$ for all $A,B \in \borel$. Additionally, we define $P_{s,F} \in \Delta([0,1])$ to be the bid distribution induced by the bidding strategy~$s$ and the value distribution $F$, i.e., $P_{s,F}(B) = \E_{v\sim F}[\kappa_s(B, v)]$ for all $B \in \borel$.

Simultaneously, the competitors submit their own bids $\{b_i\}_i \subset [0,1]$. We use $h \coloneqq \max_i b_i$ to denote the \emph{highest competing bid}, and $H \in \Delta([0,1])$ to denote its distribution (represented by its cumulative distribution function). We assume that $H$ is independent of the value-bid joint distribution $J_{s,F}$. In sealed-bid auctions, the independence follows directly from the independent-private-values assumption that prevails in much of the work on first-price auctions (e.g., see \citealt{krishna2009auction, milgrom2004putting, balseiro2023contextual, feng2021convergence}). In practice, it holds in scenarios where correlation in values is caused by a publicly-observable context, which yields independence upon conditioning, i.e., the values are independent for any fixed context. For example, in online advertising, buyers' values depend on  user-specific features, which are communicated to the buyers (or their autobidders) before bids are solicited, and act as the context for the auction. As buyers can specify different bids for different user segments, via separate campaigns if necessary, their values are often independent for each segment. Crucially, it allows us to endow the buyer with independent private information, which can be used to determine her bid but cannot be exploited by others. The ``amount'' of such information turns out to have a significant impact on performance.

Once the bids are submitted, the allocation and payment are decided according to a first-price auctions: the good is allocated to the buyer if and only if she is the highest bidder, i.e., $b \geq h$, in which case the buyer pays the auctioneer her bid $b$. If the product is not allocated to the buyer, i.e., $b < h$, the buyer does not pay anything. We posit that the buyer has a quasi-linear utility, i.e, for value $v$, an associated bid $b$, and a highest competing bid $h$, the utility of the buyer is
\begin{equation*}
    u(b,h ; v) := \left( v - b \right) \cdot \mathbbm{1} \left\{ b \geq h \right\}.
\end{equation*}

Given a bidding strategy $s$, which maps each value $v \in [0,1]$ to a distribution of bids $s(v) \in \Delta([0,1])$, and a distribution of highest competing bids $H$, the expected utility of the buyer is
\begin{equation*}
\U[F]{s, H}\ \coloneqq\   \mathbb{E}_{(v, h) \sim F \times H}  \left[ \mathbb{E}_{b \sim s(v)} \left[ u(b,h; v) \right] \right]\,.
\end{equation*}
We abuse notation slightly and use $\U[F]{s, h}$ to denote $\U[F]{s,\delta_h}$.

\noindent \textbf{Objective.} The buyer aims to select a strategy $s \in \Scal$ which maximizes her expected utility $\U[F]{s, H}$. However, as is often the case in practice, the buyer does not know the distribution of the highest competing bids $H$. In light of this uncertainty about $H$, it is natural to design bidding strategies that guarantee strong performance simultaneously against \emph{all} potential highest bid distributions $H \in \Delta([0,1])$, which is our aim in this work. This motivates us to measure the performance of a bidding strategy using regret, which is defined as
\begin{equation}
    \label{eq:regret}
    \reg_F(s,H)\ \coloneqq\  \sup_{s' \in \Scal}\ \U[F]{s',H} - \U[F]{s,H} \,.
\end{equation}

The regret $\reg_F(s, H)$ quantifies the sub-optimality of employing a given bidding strategy $s$ against the highest competing bid distribution $H$. It is defined as the difference between the utility achieved by an oracle, who knows the distribution $H$ and selects the optimal bidding strategy, and the utility obtained by our chosen bidding strategy. As the distribution $H$ is unknown and unavailable while designing $s$, we take the robust-optimization approach and aim to minimize this sub-optimality uniformly over all highest competing bid distributions $H$, i.e., we aim to minimize the \emph{worst-case regret} $\mathrm{WReg}_F(s) \coloneqq \sup_H \reg_F(s,H)$. Formally, our goal is to characterize bidding strategies which minimize worst-case regret:
\begin{equation}
\label{eq:minimax_strategy_problem}
  \inf_{s \in \Scal}\ \mathrm{WReg}_F(s)\ =\  \inf_{s \in \Scal}\ \sup_{H \in \Delta([0,1])} \reg_F(s, H)\,.
\end{equation}
When the problem \eqref{eq:minimax_strategy_problem} admits a minimizer $s^*$, we refer to it as a \emph{minimax-optimal bidding strategy}.

We conclude with a brief discussion of the model. First, note that our definition of utility assumes ties are broken in favor of the buyer under consideration. We make this choice purely for notational convenience and it is without loss of generality: the minimax-optimal bidding strategies we design under this assumption continue to be minimax optimal for all possible tie-breaking rules, as we show in \Cref{appendix:tie-breaking}. Intuitively, this is because our minimax-optimal strategies induce absolutely continuous bid distributions, thereby making ties a zero-probability event.

Next, note that our definition of regret in \eqref{eq:regret} is at the ex-ante stage, i.e., regret is measured in expectation over the private value of the buyer. This choice is motivated by online display advertising, where a buyer (advertiser) typically participates in thousands of first-price auctions as a part of their ad campaign. Therefore, standard concentration arguments apply, and any strategy which does well in expectation ends up performing well cumulatively across the large number of auctions. Similar reasoning has motivated prior works on budget management in auctions to use expected utility as the objective and study budget constraints which hold in expectation. For example, see \citet{gummadi2012repeated, abhishek2013optimal, balseiro2021budget, balseiro2023contextual} for models of single-shot auctions with in-expectation constraints and objectives. Thus, bidding strategies with low worst-case ex-ante regret would yield good performance over the entire campaign. Especially in uncertain market conditions and high-volatility periods that make it impossible to use machine-learning techniques to learn good strategies, either due to a dearth of data or rapid changes in the market that render past data obsolete. Our work offers a robust alternative to learning-based methods: minimax-optimal bidding strategies come with regret guarantees which hold regardless of how the market behaves, and hold for each auction individually.

Furthermore, our model treats the value distribution $F$ as a parameter and allows it to take arbitrary values. One particular value it can take is $\delta_v$, which corresponds to the case where the value is deterministic and equal to $v$. Thus, our model can also capture the interim regret minimization problem, where the value $v$ of the buyer is fixed and known, and she wishes to minimize worst-case regret over all possible highest-competing bid distributions. In other words, our model is more general that the one which measures regret at the interim stage. This added generality is crucial because it endows the buyer with independent private information that can be used by the buyer but not the competition. Such information is common in real-life auctions due to idiosyncratic preferences of the participants. The amount of variation in the value distribution is a measure of the ``amount'' of private information, and understanding its impact on regret is one of our primary contributions. However, the more general definition of regret comes with a cost: it makes our analysis significantly more challenging because our problem is parameterized by an infinite-dimensional value distribution $F \in \Delta([0,1])$, instead of the one-dimensional  value $v$.

\section{Performance Evaluation}\label{sec:strat_evaluation}

A crucial first step in the design of good bidding strategies is the ability to evaluate their performance on the metric of interest, which for us is the worst-case regret criterion defined in \Cref{sec:model}. However, computing the worst-case regret $\mathrm{WReg}_F(s)$ for a bidding strategy is challenging because it requires one to solve an infinite-dimensional maximization problem over the set of highest-competing-bid distributions $H$. Furthermore, the benchmark $\sup_{s' \in \Scal} \U[F]{s',H}$ against which we measure regret is itself the value of an optimization problem over $\Scal$—the space of all bidding strategies which map values to distributions. Consequently, $\reg_F(s,H)$ is itself the value of an infinite-dimensional optimization problem. These challenges are further exacerbated by the fact that all these optimization problems are parameterized by the value distribution $F$, which can be an arbitrary element in the space of all distributions $\Delta([0,1])$.

Our first main result reduces this complicated infinite-dimensional optimization to a simple one-dimensional optimization. It does so for the large class of bidding strategies comprised of all $s \in \mathcal{S}$ which induce a continuous bid distribution $P_{s,F}$.

\begin{theorem}\label{thm:evaluation}
	For any value distribution $F \in \Delta([0,1])$ and any bidding strategy $s \in \Scal$ which induces an absolutely continuous bid distribution $P_{s,F}$, we have
    \begin{align*}
		\mathrm{WReg}_F(s)= \sup_{H \in \Delta([0,1])} \reg_F(s,H)\ =\ \sup_{h\in [0,1]} \reg_F(s,\delta_h)= \sup_{h\in [0,1]}\E_{v \sim F}\left[ (v-h) \cdot \mathbbm{1}(v \geq h) \right]-\ \U[F]{s,h} \,.
    \end{align*}
    \normalsize
\end{theorem}

\Cref{thm:evaluation} is a critical result that characterizes the worst-case regret $\mathrm{WReg}_F(s)$ of a strategy as the output of a one-dimensional optimization over deterministic highest competing bids.
It implies that, when evaluating the worst-case performance over all highest-competing-bid distributions $H$, the buyer may restrict attention to deterministic highest competing bids, i.e., distributions of the form  $H = \delta_h$. Importantly, deterministic highest competing bids drastically simplify the evaluation of the benchmark $\sup_{s' \in \Scal} \U[F]{s',\delta_h}$: the optimal strategy for the benchmark is to bid $h$ whenever the value is higher than $h$. This yields the simplified expression $\E_{v \sim F}\left[ (v-h) \cdot \mathbbm{1}(v \geq h) \right]$ for the benchmark in \Cref{thm:evaluation}. Crucially, the simplification of \Cref{thm:evaluation} holds for all value distributions $F$ and bidding strategies, so long as they induce an absolutely continuous bid distribution $P_{s,F}$. In which case, the worst-case regret can be computed via a simple line search.

The proof of \Cref{thm:evaluation} follows from Bauer Maximum Principle (see 7.69 of \citealt{aliprantis2006infinite}). However, the application of such a result is not straightforward in our setting as the function $H \mapsto \reg_F(s,H)$ is not continuous on $\Delta([0,1])$ for the weak topology. Indeed, for a fixed bid $b$ and value $v > b$, note that the utility function is not continuous as a function of the highest competing bid $h$. As such, $\reg_F(s,H)$ is the pointwise supremum over an infinite dimensional space (over the bidding strategies chosen by the oracle) of a function which is discontinuous. To get around this challenge, we first use the Portmonteau lemma to establish the upper semi-continuity of the expected utility, and then use a version of Berge's maximum theorem for upper semi-continuous functions to conclude that $\reg_F(s,H)$, which is the pointwise supremum of such functions, is also upper semi-continuous.

We conclude this section by demonstrating the benefits of \Cref{thm:evaluation} through the evaluation of \emph{uniform-bid-shading strategies}, which are bidding strategies that uniformly scale each value $v$ by the same constant $\alpha \in [0,1]$ to determine the corresponding bid $\alpha \cdot v$. These strategies naturally arise in online ad auctions, they have been documented in laboratory experiments~\cite{bajari2005structural, filiz2007auctions}, and have been the popular choice for ``simple'' strategies in prior work~\citep{fikioris2024liquid, gaitonde2023budget}. Crucially, if one were to treat each value in isolation and minimize regret one value at a time, as \citet{kasberger2023robust} did, then the optimal deterministic strategy ends up being the uniform-bid-shading policy $v \mapsto 0.5 \cdot v$ with $\alpha = 0.5$. Note that separately evaluating worst-case regret for each value $v$ assumes that the choice of the highest competing bid distribution $H$ can depend on the value, which is tantamount to assuming that the buyer's value is not private information. This begs the following questions: Do uniform-bid-shading strategies work well if one accounts for private information in the form of a non-deterministic value distribution $F$? If so, what should be the shading factor $\alpha$?

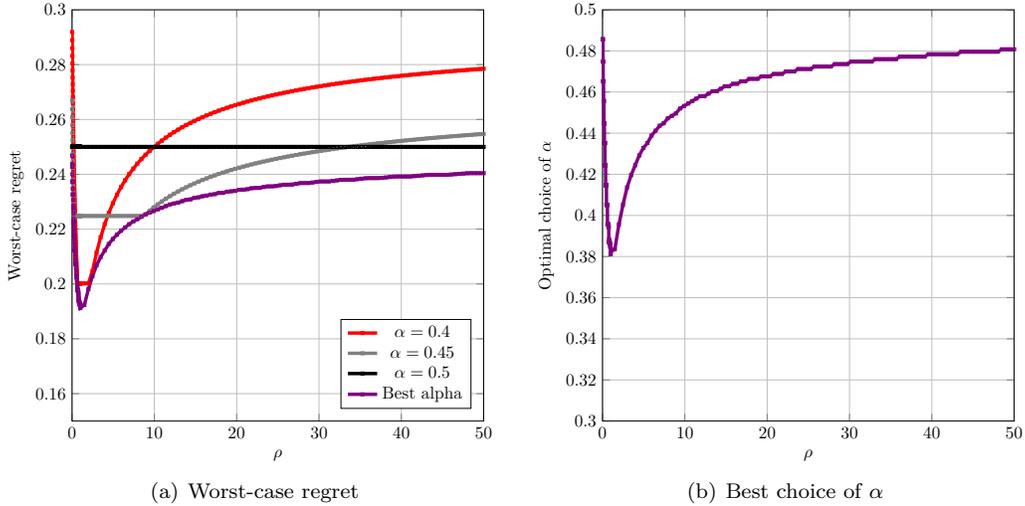
\begin{figure}[h!]
    \centering
    \subfigure[Worst-case regret]{
    \begin{tikzpicture}[scale = 0.65]
    \begin{axis}[
        width=10cm,
        height=10cm,
        xmin=0,xmax=50,
        ymin=0.15,ymax=0.3,
        table/col sep=comma,
        xlabel={$\rho$},
        ylabel={Worst-case regret},
        grid=both,
        legend pos=south east,
    ]

    \addplot [red,  line width=0.7mm, ,mark=square,mark options={scale=.1}] table[x=beta,y={alpha = 0.4004}] {Data/worst_case_regret_all_betas.csv};
    \addlegendentry{$\alpha = 0.4$}

    \addplot [gray,  line width=0.7mm, ,mark=square,mark options={scale=.1}] table[x=beta,y={alpha = 0.4496}] {Data/worst_case_regret_all_betas.csv};
    \addlegendentry{$\alpha = 0.45$}

    \addplot [black,  line width=0.7mm, ,mark=square,mark options={scale=.1}] table[x=beta,y={alpha = 0.5}] {Data/worst_case_regret_all_betas.csv};
    \addlegendentry{$\alpha = 0.5$}

    \addplot [violet,  line width=0.7mm, ,mark=square,mark options={scale=.1}] table[x=beta,y={best_performance}] {Data/best_alpha_and_performance_beta_dist.csv};
    \addlegendentry{Best alpha}

    \end{axis}
    \end{tikzpicture}
    \label{fig:regret_bid_shade}
    }
    \subfigure[Best choice of $\alpha$]{
    \begin{tikzpicture}[scale = 0.65]
    \begin{axis}[
        width=10cm,
        height=10cm,
        xmin=0,xmax=50,
        ymin=0.3,ymax=0.5,
        table/col sep=comma,
        xlabel={$\rho$},
        ylabel={Optimal choice of $\alpha$},
        grid=both,
        legend pos=south east,
    ]

    \addplot [violet,  line width=0.7mm, ,mark=square,mark options={scale=.1}] table[x=beta,y={best_alpha}] {Data/best_alpha_and_performance_beta_dist.csv};

    \end{axis}
    \end{tikzpicture}
    \label{fig:best_alpha}
    }
    \caption{\textbf{Performance of uniform-bid-shading strategies for Beta$(\rho,\rho)$ value distributions.} In (a), we report the worst-case regret of different uniform-bid-shading strategies $s_{\alpha}$ as a function of $\rho$, when the value distributions is of the form Beta$(\rho,\rho)$. We also report the worst-case regret for the uniform-bid-shade strategy which uses the best $\alpha$ for each $\rho$ (see violet curve). In (b), we plot the best choice of $\alpha$ for each value of $\rho$.}
    \label{fig:alpha_bid_shade}
\end{figure}

To answer these questions, we leverage \Cref{thm:evaluation} to evaluate the worst-case regret of uniform-bid-shading strategies across various value distributions, where each evaluation reduces to a simple line search over $h \in [0,1]$. In \Cref{fig:alpha_bid_shade}, we plot the worst-case regret of uniform-bid-shading strategies when the value distribution follows a $\mathrm{Beta}(\rho,\rho)$ distribution for $\rho > 0$. This provides a large class of value distributions for comparison, each with a mean of 0.5 but different levels of spread in values. \Cref{fig:alpha_bid_shade} highlights that, when the competing bids can only depend on the value distribution but not the realized value, the optimal choice of $\alpha$ depends heavily on the buyer's value distribution. When $\rho$ is very high, or very low, the choice $\alpha =0.5$ suggested by \cite{kasberger2023robust} performs strongly among uniform-bid-shading strategies. This aligns with their findings, as the value distribution converges to the point mass at $0.5$ for high $\rho$ and a Bernoulli distribution (which is essentially a point mass at $1$ for our problem), when $\rho$ is low. However, when the value distribution is more dispersed---for instance when $\rho = 1$ (corresponding to the uniform distribution)---worst-case regret can be significantly reduced by adopting the uniform-bid-shading strategy with $\alpha = 0.4$. In fact, this choice results in a worst-case regret that is $20\%$ lower than than that of $\alpha = 0.5$.  This improvement highlights the importance of considering the entire value distribution when designing bidding strategies, particularly in practical settings where the buyer's realized value is private and not observed by competitors.

All in all, \Cref{fig:alpha_bid_shade} demonstrates that \Cref{thm:evaluation} provides a powerful framework for comparing bidding strategies on the basis of worst-case regret. It equips the buyer with a principled method for selecting the strategy best suited to their value distribution. However, this approach is limited to evaluating and selecting among a predefined set of strategies, and leaves open the question: can we go beyond simple bidding strategies and characterize the optimal one that minimizes worst-case regret? In the next section, we answer this question in the affirmative by providing an efficient procedure for constructing minimax-optimal bidding strategies for arbitrary value distributions.

\section{Minimax-Optimal Bidding Strategy}\label{sec:minimax}

Directly computing a minimax-optimal strategy by solving \eqref{eq:minimax_strategy_problem} is inherently challenging. Beyond the challenges discussed in \Cref{sec:strat_evaluation} for evaluating the worst-case regret  $\mathrm{WReg}_F(s)$ of a fixed strategy $s$, we now face the additional challenge of minimizing worst-case regret over randomized bidding strategies $s \in \Scal$. Each such strategy defines an uncountable family of distributions $\{s(v) \in \Delta([0,1]) \mid v \in [0,1] \}$ over bids, further complicating the problem. 
In order to address the complexities arising from these nested infinite-dimensional optimization problems, we proceed in multiple steps, each aimed at simplifying the overall minimax optimization to make it tractable. 

In \Cref{sec:full-info}, we build on the reduction established in \Cref{thm:evaluation} to reformulate the worst-case regret problem as an alternative maximization problem. This new formulation leverages a full-information benchmark with a significantly simpler utility expression than the one used in \eqref{eq:regret}.
This result allows us to considerably simplify the inner maximization problem of \eqref{eq:minimax_strategy_problem}, by simplifying both the space of distributions of highest competing bids and the expression of the regret objective.
 In \Cref{sec:quantile-strat}, we introduce the class of quantile-based bidding strategies which specify a single bid for each \emph{quantile}, instead of a distribution of bids for each value. The resulting re-parametrization significantly simplifies the space of bidding strategies 
 and allows us to represent bidding strategies using a single function from $[0,1]$ to itself---a stark contrast to the uncountable family of bid distributions required to specify a general bidding strategy.
In \Cref{sec:saddle_point}, we leverage this simplification to construct a quantile-based bidding strategy $Q^*$ and a highest-competing-bid distribution $H^*$ which form a candidate saddle-point for our problem. The construction is based on ordinary differential equations that capture appropriate first-order-optimality conditions.
In \Cref{sec:main-result}, we establish our main result (\Cref{thm:main-result}) by proving  that our candidate $(Q^*,H^*)$ does in fact form a saddle-point for \eqref{eq:minimax_strategy_problem}. 

Throughout this section, we make the assumption that the value distribution satisfies $F(0) = 0$, i.e., the probability of the value being zero is zero. It is made purely for ease of exposition, and all our results readily extended to the general case, as we show in \Cref{appendix:atom-at-0}.

\subsection{Reduction to the Full-Information Benchmark}\label{sec:full-info}

Recall that the regret $\reg_F(s,H)$ of a bidding strategy $s$ against the highest-competing-bid distribution $H$ measures its utility loss against the benchmark $\sup_{s'} \U[F]{s',H}$. This benchmark is based on partial information and captures the optimal utility that can be attained with knowledge of the distribution $H$, but not the realization of the highest competing bid $h$. One can alternatively consider the stronger full-information benchmark that also has knowledge of the highest competing bid $h$ that was realized. In particular, for a given highest-competing-bid distribution $H$, the \emph{full-information benchmark} is defined as
\begin{align*}
	\O_F(H)\ \coloneqq\ \E_{h \sim H}\left[ \max_{s' \in \Scal}\  \U[F]{s', h} \right].
\end{align*}
 Observe that, if the highest competing bid $h$ is known, then the utility-maximizing strategy is to simply bid $h$ whenever the value $v$ is greater than or equal to it. Thus, the full-information benchmark can be rewritten as
\begin{align*}
	\O_F(H)\ =\ \E_{h \sim H}\left[ \max_{s' \in \Scal}\  \U[F]{s',h} \right] = \E_{(v,h) \sim F \times H}\left[ (v-h) \cdot \mathbbm{1}(v \geq h) \right]\,.
\end{align*}
Based on the full-information benchmark, one can define an alternative notion of regret
\begin{align*}
	\rr[F]{s,H}\ \coloneqq \O_F(H)\ -\ \U[F]{s,H}\,. 
\end{align*}

With a slight abuse of notation, we use $\O_F(h)$ and $\rr[F]{s,h}$ to represent $\O_F(\delta_h)$ and $\rr[F]{s,\delta_h}$ respectively, where $\delta_h$ is the Dirac distribution. The following result follows from \Cref{thm:evaluation} and show that for every bidding strategy $s$ which induces an absolutely continuous bid distribution $P_{s,F}$, the worst-case regret is identical under both partial-information and full-information benchmarks. 
\begin{corollary}\label{thm:full-info}
	For any value distribution $F \in \Delta([0,1])$ and any bidding strategy $s \in \Scal$ which induces an absolutely continuous bid distribution $P_{s,F}$, we have
	\begin{align*}
		\sup_{H \in \Delta([0,1])} \reg_F(s,H)\ =\ \sup_{h \in [0,1]}\ \rr[F]{s, h}\ =\ \sup_{H \in \Delta([0,1])}\ \rr[F]{s, H}\,.
	\end{align*}
\end{corollary}
\Cref{thm:full-info} directly follows from \Cref{thm:evaluation} by noting that, for every $h \in [0,1]$, we have 
\begin{equation*}
    \rr[F]{s,h} = \E_{v \sim F}\left[ (v-h) \cdot \mathbbm{1}(v \geq h) \right]-\ \U[F]{s,\delta_h}.
\end{equation*}
Moreover, we remark that for every $H \in \Delta([0,1])$ we have that $\rr[F]{s,H} = \E_{h \sim H}[\rr[F]{s,h}]$. Therefore, $\sup_{h \in [0,1]}\ \rr[F]{s, h}\ =\ \sup_{H \in \Delta([0,1])}\ \rr[F]{s, H}$, which establishes the second equality.

\Cref{thm:full-info} has a powerful interpretation, as it implies that the robust value of the information associated with the knowledge of the \emph{distribution} of the highest competing bid is equal to the one associated with knowledge of its \emph{realization}.
In what follows, we use \Cref{thm:full-info} to reduce the problem of finding a minimax-optimal bidding strategy for the partial-information benchmark defined in \eqref{eq:minimax_strategy_problem} to that of finding one for of the simpler full-information benchmark. Namely, we now aim at solving
\begin{equation}
    \label{eq:minimax_full_info}
    \inf_{s \in \mathcal{S}} \sup_{h \in [0,1]} \rr[F]{s,h}.
\end{equation}
We note that this reduction simplified not only the optimization space in the inner supremum but also the objective itself. 

\subsection{Quantile-Based Strategies}\label{sec:quantile-strat}

Next, we further simplify the minimax problem \eqref{eq:minimax_full_info} by reducing the strategy space of the buyer from the set $\Scal$ of arbitrary randomized strategies $s: [0,1] \to \Delta([0,1])$ to simpler \emph{quantile-based strategies}, with the latter being parameterized by a single function $Q:[0,1] \to [0,1]$. The primary goal of this section is to describe this reformulation and build intuition for it. Accordingly, we first provide informal arguments to motivate our reformulation, before formally defining quantile-based bidding strategies in \Cref{def:quantile-based-bidding-strategy}.

Note that, for any fixed highest-competing-bid distribution $H$, the regret-minimization problem
\begin{align*}
	\inf_{s \in \Scal}\ \rr[F]{s,H}\ =\ \inf_{s \in \Scal}\ \O_F(H)\ -\ \U[F]{s,H}\ =\ \O_F(H)\ -\ \sup_{s \in \Scal}\ \U[F]{s,H}
\end{align*}
boils down to the utility-maximization problem $\sup_{s \in \Scal}\ \U[F]{s,H}$ for the buyer. And therefore, we can leverage the following intuitive observation about utility-maximizing strategies in first-price auctions: higher values result in higher optimal bids. In other words, the optimal-bid correspondence
\begin{equation*}
    v \mapsto {\mathbb S(v) \coloneqq \argmax_{b \in [0,1]} \ (v-b)\cdot H(b)}
\end{equation*}
is increasing in the value $v$, i.e., $\mathbb S(v_1) \succcurlyeq \mathbb S(v_2)$ for all $v_1 \geq v_2$. 
The proof of this monotonicity result is a standard application of monotone comparative statics and various forms of it have been proven in prior work (see, for example, \citealt{kumar2024strategically}). We do not directly use it in our proof, but instead treat it as a guiding principle for our reformulation.

Let $\ell(\mathbb S(v)) = \sup \mathbb S(v) - \inf \mathbb S(v)$ denote the ``length'' of the set of optimal bids for value $v$, and observe that $\sum_{v \in [0,1]} \ell(\mathbb S(v)) \leq 1$. As the sum of any uncountable collection of positive integers necessarily diverges to infinity, we must have that the set $\{ v \mid \ell(\mathbb S(v))> 0\}$ is at most countable, and the remaining values have a unique optimal bid. If $F$ had no atoms, the probability of $v \sim F$ taking one of these countably-many values would be 0, and we could focus on deterministic bidding strategies for the remaining values. However, when $F$ contains atoms, we do not have this luxury. The following example illustrates this fact by showin that prior-independent pricing~\citep{bergemann2011robust} is a special case of our problem with a deterministic value.
\begin{example}\label{example:atom-at-1}
	Suppose the value is deterministic and equal to $1$. As there is only one value, the bidding strategy is simply the distribution of bids  $s(1)$ for value 1. For the highest-competing-bid distribution $H$, the expected regret is given by
    \begin{align*}
        \rr[F]{s, H}\ =\ \E_{h \sim H}[1 - h]\ -\ \E_{h \sim H, b \sim s(1)}\left[ (1 - b) \mathbbm{1}(b \geq h)\right]\,. 
    \end{align*}
    To see the equivalence to prior-independent pricing, define $\nu = 1-h$ to be the willingness-to-pay (WTP) and $p = 1 - b$ to be the price. Then, regret can be re-written as
    \begin{align*}
         \rr[F]{s, H}\ =\ \E_{(1- \nu) \sim H}[\nu]\ -\ \E_{(1 -\nu) \sim H, (1 -p) \sim s(1)}\left[ p \mathbbm{1}(\nu \geq p)\right]\,.
    \end{align*}
    The right-hand side is exactly the expected regret for the pricing problem when the WTP $\nu$ is distributed according to $(1 - \nu) \sim H$ and the price $p$ is distributed according to $(1 - p) \sim s(1)$. Thus, computing the minimax regret over $s$ and $H$ is equivalent to computing the minimax regret over all randomized pricing policies and WTP distributions. Hence, these two problems are equivalent and the sub-optimality of deterministic pricing policies established in \citet{bergemann2011robust} implies the sub-optimality of deterministic bidding strategies in this example. 
\end{example}

Fortunately (and surprisingly), it turns out that there are no other cases requiring randomization: we can focus on bidding strategies which bid randomly only for atoms of $F$, and bid deterministically for all other values. This is where we make our second crucial observation: quantiles provide a natural method to specify bids which are random for atoms of $F$, but deterministic for all other values. In particular, it is well known that $v = F^-(y)$ is distributed according to $F$ when the quantile $y \in [0,1]$ is uniformly distributed. Importantly, there are infinitely many quantiles $y \in [0,1]$ which result in $v = F(y)$ whenever $v$ is an atom, and there is exactly one such quantile for all other values; exactly as desired. Thus, if we specify a bid for each quantile $y \in [0,1]$, instead of each value, we naturally end up with a bidding strategy which is random for atoms and deterministic otherwise.

\begin{definition}\label{def:quantile-based-bidding-strategy}
	We refer to any strictly-increasing absolutely-continuous function ${Q: [0,1] \to [0,1]}$ as a quantile-based bidding strategy, and use $\mathcal Q$ to refer to the set of all such functions. For every $Q \in \mathcal Q$, we define the corresponding bidding strategy $\s[Q] \in \Scal$ as follows:
	\begin{align*}
		\s[Q](v)\ = \begin{cases}
			\delta_{Q(F(v))} &\text{if } \lambda(Y(v)) = F(\{v\}) = 0\\
			\{Q(y)\mid y \sim \unif(Y(v))\} &\text{if } \lambda(Y(v)) = F(\{v\}) >  0
		\end{cases}
	\end{align*}
	where $Y(v) \coloneqq \{y \in [0,1] \mid F^-(y) = v\}$ is the interval of quantiles which correspond to the value $v$.
\end{definition}

\begin{remark}\label{remark:Y-interval}
	Note that $Y(v)$ is the pre-image of $\{v\}$ for the non-decreasing function ${F^-: [0,1] \to [0,1]}$, and as a consequence it is an interval in $[0,1]$. Also, if $y \sim \unif(0,1)$, then $F^-(y) \sim F$~\citep{embrechts2013note}. Therefore, $\lambda(Y(v)) = F(\{v\})$, as noted in \Cref{def:quantile-based-bidding-strategy}.
\end{remark}

\begin{remark}\label{remark:kernel}
	$\kappa_{\s[Q]}$ is a Markov kernel and $\s[Q] \in \Scal$. See \Cref{appendix:minimax-remark} for details. 
\end{remark}

Hereafter, we focus our search for minimax-optimal policies to quantile-based bidding strategies $\mathcal Q$. In particular, we analyze the following regret-minimization problem:
\begin{equation*}
    \inf_{Q \in \mathcal{Q}}\ \sup_{h \in [0,1]}\ \rr[F]{\s[Q],h}
\end{equation*}

 Since the highest possible bid under $\s[Q]$ is $Q(1)$, setting $h > Q(1)$ only hurts the benchmark and not the utility of the buyer. In other words $\rr[F]{\s[Q],h} \leq \rr[F]{\s[Q],Q(1)}$ for all $h > Q(1)$. Thus, we can focus on $h \in [0,Q(1)]$. Alternatively, we can focus on $h = Q(y)$ for some $y \in[0,1]$. In this case, the benchmark can be written as
\begin{align*}
	\mathcal{O}_{F}(h)\ =\ \E_{v \sim F}\left[\mathcal(v-h)\mathbbm{1}(v \geq h) \right]\ =\ \int_{h}^1 (1 - F(t)) dt\ =\ \int_{Q(y)}^1 (1 - F(t)) dt
\end{align*}
and the expected utility of $\s[Q]$ can be written as
\begin{align*}
	\U[F]{\s[Q], h}\ =\ \E_{v \sim F}[ \E_{b \sim \s[Q](v)} [(v - b) \mathbbm{1}(b \geq h)]]\,.
\end{align*}

\begin{lemma}\label{lemma:quantile-based-bidding}
	For every quantile-based bidding strategy $Q \in \mathcal{Q}$ and highest competing bid $h = Q(y)$ with $y \in [0,1]$, we have
	\begin{align*}
		\U[F]{\s[Q], h}\ =\ \int_y^1 (F^-(t) - Q(t))dt\,.
	\end{align*}
    Moreover, the distribution of bids induced by $\s[Q]$ (and $F$) is $P_{\s[Q], F} = \{Q(t) \mid t \sim \unif(0,1)\}$.
\end{lemma}

We note that the joint distribution over values and bids is not a product distribution in general as, for any practical bidding strategy, the bid submitted by the buyer depends on the value. 
\Cref{lemma:quantile-based-bidding} allows us to express the expected utility with respect to the intricate joint distribution as a much simpler expectation over quantiles. In particular, the right-hand side is the expected utility under the distribution which ``couples'' the values $v = F^-(t)$ and the bids $Q(t)$ by generating them from the same uniform distribution $t \sim \unif(0,1)$.

\Cref{lemma:quantile-based-bidding} allows us to reduce the minimax problem to the following simple form
\begin{equation}
\label{eq:simpler-saddle}
	\inf_{Q \in \mathcal{Q}}\ \sup_{y \in [0,1]}\ \int_{Q(y)}^1 (1 - F(t)) dt\ -\ \int_y^1 (F^-(t) - Q(t)) dt \,.
\end{equation}
Observe that restricting our attention to quantile-based strategies considerably simplifies the buyer's strategy space. Quantile-based bidding strategies can be specified with a single monotonic function $Q:[0,1] \to [0,1]$, whereas general bidding strategies $s: [0,1] \to \Delta([0,1])$ require an uncountable collection of monotonic functions, namely the CDFs of the bid distributions $\{s(v)\}_v$ for each value $v$. 
Furthermore, re-expressing the regret as a function of the quantile-based bidding strategy $Q$ allows us to obtain a formulation that is convex in $Q$ for any fixed $y$. Indeed, the function
\begin{align*}
	Q \longmapsto \int_{Q(y)}^1 (1 - F(t)) dt\ -\ \int_y^1 (F^-(t) - Q(t)) dt
\end{align*}
is convex, and taking a supremum over $y$ preserves this convexity. As such, the outer minimization problem \eqref{eq:simpler-saddle} corresponds to the minimization of a convex functional. In particular, local minima are also global minimax for convex functions, and the former can be found using first-order conditions---we leverage this observation in \Cref{sec:saddle_point} to construct a candidate minimax-optimal quantile-based bidding strategy.

\subsection{Saddle Point}\label{sec:saddle_point}

In light of the reformulation of the previous section, we now focus on solving the simplified minimax problem \eqref{eq:simpler-saddle}.
We do so by explicitly constructing a quantile-based bidding strategy $Q^*$ and a distribution of highest competing bids $H^*$ such that the pair $(Q^*, H^*)$ forms a saddle point, i.e.,
\begin{align*}
	\rr[F]{\s[Q^*], H^*} \leq \rr[F]{s, H^*} \quad \forall\ s \in \mathcal{S} \quad \text{and} \quad \rr[F]{\s[Q^*], H^*} \geq \rr[F]{\s[Q^*],H} \quad \forall\ H \in \Delta([0,1])\,.
\end{align*}

Intuitively, the minimax problem in \eqref{eq:simpler-saddle} is a convex-concave saddle-point problem and we would like to leverage this convex structure. However, as $\mathcal Q$ was defined to be the set of all strictly-increasing absolutely-continuous functions, it is not compact under standard topologies. This makes a direct application of results from infinite-dimensional convex optimization difficult. 
But all is not lost, the convex structure of the problem lends itself to a constructive proof using first-order optimality conditions. We leverage this fact to construct a candidate quantile-based bidding strategy $Q^*$ in \Cref{sec:construct_Q} and a candidate distribution of highest competing bids $H^*$ in \Cref{sec:construct_H}. 

\subsubsection{Candidate quantile-based bidding strategy}
\label{sec:construct_Q}

Fix a quantile-based bidding strategy $Q$ and consider the inner-maximization problem over quantiles $y \in [0,1]$ in \eqref{eq:simpler-saddle}:
\begin{align}
\label{eq:inner_max}
	\sup_{y \in [0,1]}\quad \int_{Q(y)}^1 (1 - F(t)) dt\ -\ \int_y^1 (F^-(t) - Q(t)) dt \,.
\end{align}

Looking forward, we would like all $y \in [0,1]$ to be optimal and have the same objective value. If all $y \in [0,1]$ have the same objective value, the derivative with respect to $y$ must be 0, which yields the following first-order condition:
\begin{align}\label{eq:first-ode}
	- (1 - F(Q(y))) \cdot Q'(y) \ +\ (F^-(y) - Q(y)) = 0 &&\forall\ y \in [0,1]\,.
\end{align}

This ODE will form a crucial part of our constructive characterization of the saddle point $(Q^*, H^*)$. 
Our next result establishes that it admits a well-behaved solution.

\begin{lemma}\label{lemma:ode-existence}
	For every value distribution $F$, there exists an absolutely continuous function\\ ${Q^*: [0,1] \to [0,1]}$ such that the following conditions hold:
	\begin{enumerate}
		\item First-order optimality: There exists a measurable set $A \subseteq [0,1]$ with measure $\lambda(A) = 1$ such that $Q^*$ is differentiable and $ Q^{*'}(y) \ =\ (F^-(y) - Q^*(y))/(1 - F(Q^*(y)))$  for all $y \in A$.
		\item Invertibility: $Q^*$ is strictly increasing, i.e., $Q^*(y_1) > Q^*(y_2)$ for all $y_1 > y_2$.
		\item Strict Boundedness: $0 < Q^*(y) < F^-(y)$ for all non-zero $y \in (0,1]$ and $Q^*(0) = 0$.
	\end{enumerate}
\end{lemma}
\Cref{lemma:ode-existence} shows the existence of a quantile-based bidding strategy $Q^*$ which satisfies the first order condition \eqref{eq:first-ode} almost surely; it will serve as our candidate optimal quantile-based bidding strategy. 
\Cref{lemma:ode-existence} is one of our primary technical results; we provide an informal sketch of it here. 
Consider the ODE in \eqref{eq:first-ode} written in standard form,
\begin{align}
    \label{eq:main-first-ode-rewrite}
    x'(t) = g(t,x(t)) \quad \text{where}\quad g(t,x) \coloneqq \frac{F^-(t) - x}{1 - F(x)}\,, && \forall \,(t,x) 
    \in [0,1] \times [0,F^{-}(1))\,,
\end{align}
with the initial condition $x(0) = 0$. The denominator of $g$ raises the first challenge: as $x$ approaches $F^-(1)$, the denominator diverges to infinity and it is undefined beyond that point.
To the best of our knowledge, no existence theorems directly apply to such a setting. To address it, we leverage our first intuitive observation: whenever $x(t)$ approaches the $F^-(t)$ curve, the numerator of $g$ goes to zero and reduces its rate of approach $x'(t)$. This motivates us to conjecture the existence of a solution with $x(0) =0$ and $x(t) < F^-(t)$ for all $t > 0$. To prove it, we use a constant $\alpha > 0$ to modify $g$ and bound it from above:
\begin{equation*}
    x'(t) = g_\alpha(t,x(t)) \quad \text{where}\quad g_\alpha(t, x) = \frac{F^-(t) - x}{1 - F(\min\{x, \alpha\})}.
\end{equation*}
Observe that if, for some constant $\alpha > 0$, we are able to show the existence of a solution for this modified ODE with the property that $x(t) \leq \alpha$ for all $t \in [0,1]$, then such a solution would also be a solution to \eqref{eq:main-first-ode-rewrite}. We first establish the existence of a solution for this modified ODE. We remark that this result does not follow from standard existence theorems. Not only is the function $g_{\alpha}$ not guaranteed to be continuous, thereby ruling out the standard Peano's existence theorem, but even its restriction $x \mapsto g_\alpha(t,x)$ may be discontinuous for \emph{every} $t \in [0,1]$, thereby prima facie also ruling out Caratheodary's existence theorem. To get around this difficulty, we employ the existence theorem by \citet{biles1997caratheodory} and \citet{biles2000solvability} for discontinuous ODE which satisfy quasi-semicontinuity. It is a generalization of Caratheodary's theorem which relies on weaker continuity conditions that are satisfied by our use case.

With a solution $x_\alpha(t)$ in hand, we move onto the task of bounding it with $\alpha$. In fact we prove the stronger condition that $x_\alpha(t) < F^-(t)$ for all $t > 0$. To do so, we introduce the auxiliary ODE:
\begin{align*}
    x'(t) = \bar g(t,x(t)) \quad \text{where}\quad \bar g(t,x) \coloneqq \frac{F^-(t) \cdot (1 -x)}{1 - F(x)}\,, && \forall \,(t,x) 
    \in [0,1] \times [0,F^{-}(1))\,.
\end{align*}

Crucially, note that (i) this ODE is separable, which allows us to explicitly construct a solution $\bar x(t)$ for it, and (ii) it satisfies $\bar g(t, x) \geq g(t,x)$, which allows us to use its solution to bound the solution of the modified ODE, i.e., $x_\alpha(t) \leq \bar x(t)$. Therefore, to prove $x(t) < F^-(t)$, it suffices to show that $\bar x(t) < F^-(t)$. Intuitively, to see why this is the case, write the ODE in separable form
\begin{align*}
    \frac{1 - F(x)}{1 -x} \cdot dx\ =\ F^-(t) \cdot dt\,.
\end{align*}
Therefore, if $\bar x(z) =  F^-(z)$ for some $z > 0$, then
\begin{align*}
    \int_0^{F^-(z)} (1 - F(x)) dx\ <\ \int_0^{F^-(z)} \frac{1 - F(x)}{1 - x} dx\ =\ \int_0^{\bar x(z)} \frac{1 - F(x)}{1 - x} dx\ =\  \int_0^z F^-(t) dt\,.
\end{align*}
Integration by parts reveals the left-most term to be larger than $\E_F[v \mathbbm{1}(v \leq F^-(z))]$, which is no smaller than the right-most term. Thus, we have a contradiction and $\bar x(z) <  F^-(z)$ for all $z > 0$, as required to establish \Cref{lemma:ode-existence}. The full proof has to contend with subtle technical nuances which we have skipped here; refer to Appendix~\ref{appendix:minimax-lemmas} for details.

Next, we show that the quantile-based bidding strategy $Q^*$ described in Lemma~\ref{lemma:ode-existence} does indeed imply that all $y \in [0,1]$ are optimal solutions for the inner-maximization problem \eqref{eq:inner_max}.

\begin{lemma}\label{lemma:saddle-max-prob}
	Consider the quantile-based bidding strategy $Q^*$ defined in Lemma~\ref{lemma:ode-existence}. Then, for all $y \in [0,1]$, we have
	\begin{align*}
		\int_{Q^*(y)}^1 (1 - F(t)) dt\ -\ \int_y^1 (F^-(t) - Q^*(t)) dt\ =\ \int_0^1 Q^*(t) dt\,.
	\end{align*} 
\end{lemma}
\Cref{lemma:saddle-max-prob} concludes our construction of the candidate quantile-based bidding strategy $Q^*$. It shows that the bidding strategy defined in \Cref{lemma:ode-existence} makes nature indifferent between any highest competing bid in the interval $[0,Q^*(1)]$. As a consequence, any highest-competing-bid distribution on $[0,Q^*(1)]$ is a best-response against $Q^*$. To conclude our saddle-point construction, we construct one such particular highest-competing-bid distribution $H^*$ which in turn admits our quantile-based bidding strategy $Q^*$ as a best response.

\subsubsection{Candidate highest-competing-bid distribution}\label{sec:construct_H}

We start our construction of a candidate highest-competing-bid distribution $H^*$ by looking at the first-order optimality conditions which arise when the
buyer chooses the optimal bid $b$ for each value $v$. In contrast to the previous section, where we encountered a single first-order optimality condition, here we face an uncountable family of such conditions, one for each value, which we combine into a single ordinary differential equation (ODE) via $Q^*$.
 
Consider an absolutely continuous highest-competing-bid distribution $H$. For a fixed value $v \in [0,1]$, the expected utility of the buyer as a function of the bid $b$ satisfies
\begin{align*}
	u(b \mid v,H)\ \coloneqq\ \E_{h \sim H}[(v - b) \cdot \mathbbm{1} \{ b \geq h\}]\ =\ (v-b) \cdot H(b)\,.
\end{align*}

If $b^*(v) \in \argmax_{b \in [0,1]} u(b \mid v,H)$ is an optimal bid for value $v$, and $u(\cdot \mid v,H)$ is differentiable at $b^*(v)$, then it must satisfy
\begin{align*}
	u'(b^*(v) \mid v,H)\ =\ 0 \quad \iff \quad (v- b^*(v)) \cdot H'(b^*(v))\ -\ H(b^*(v))\ =\ 0\,.
\end{align*}
As we are seeking a distribution $H$ for which our quantile-based bidding strategy $Q^*$ is optimal, the bid $b^*(v) = Q^*(y)$ must satisfy the first-order optimality condition whenever $v = F^-(y)$. This yields the ordinary differential equation
\begin{align}\label{eq:second-ode}
	(F^-(y) - Q^*(y)) \cdot H'(Q^*(y))\ -\ H(Q^*(y))\ =\ 0\,.
\end{align}
Unlike the previous ODE \eqref{eq:first-ode}, the ODE in \eqref{eq:second-ode} is separable. We leverage this fact to characterize a solution in the following lemma. 
\begin{lemma}\label{lemma:second-ode}
	Consider the quantile-based bidding strategy $Q^*$ defined in Lemma~\ref{lemma:ode-existence}, and define\\ $H^*:[0,Q^*(1)] \to [0,1]$ as
	\begin{align*}
		H^*(x)\ \coloneqq\ \exp\left( - \int_{(Q^*)^{-1}(x)}^1 \frac{1}{1 - F(Q^*(t))} dt \right)\,.
	\end{align*}
	Then, $G \coloneqq H^* \circ Q^*$ is absolutely continuous and satisfies
	\begin{align*}
		G'(y)\ =\ G(y) \cdot \frac{1}{1 - F(Q^*(y))}\ =\ G(y) \cdot \frac{(Q^*)'(y)}{F^-(y) - Q^*(y)}\,
	\end{align*}
	almost surely for $y \in [0,1]$.
\end{lemma}
\Cref{lemma:second-ode} constructs a solution to ODE \eqref{eq:second-ode}. Our construction is obtained by solving ODE \eqref{eq:second-ode} with the initial condition $H(Q^*(1)) = 1$, and moving backwards towards $0$. The choice of this initial condition is designed to obtain a distribution $H^*$ which is supported on $[0,Q^*(1)]$. We also note that, in \Cref{lemma:second-ode}, we express the ODE in terms of $G$, as opposed to $H$. This characterization is more convenient because defining $H^*$ involves the function $(Q^*)^{-1}$, for which we have not established absolute continuity, nor has its derivative been characterized. 
We next show that the characterization of $G$ and $G'$ obtained in \Cref{lemma:second-ode} is sufficient to prove that the quantile-based bidding strategy $Q^*$ is optimal against the candidate highest-competing-bid distribution $H^*$.

\begin{lemma}\label{lemma:bid-strat-optimality}
	Consider the quantile-based bidding strategy from \Cref{lemma:ode-existence} and the highest-competing-bid distribution $H^*$ from \Cref{lemma:second-ode}; let $G = H^* \circ Q^*$. Then, for every value $v \in [0,1]$, 
    we have that
	\begin{align*}
		\{Q^*(y) \mid y \in[0,1], \, F^-(y) = v\}\ \subseteq\ \argmax_{b \in [0,1]}\ u(b \mid v, H^*)\,.
	\end{align*} 
\end{lemma}
For every value $v \in [0,1]$, \Cref{lemma:bid-strat-optimality} establishes that \emph{every} bid in the support of $\s[Q^*](v)$---which is the bid distribution for value $v$ prescribed by the quantile-based bidding strategy $Q^*$---is optimal against the highest-competing-bid distribution $H^*$ defined in \Cref{lemma:second-ode}. In other words, having shown that any distribution on $[0,Q^*(1)]$ produces the same regret against $Q^*$, we have now constructed a distribution $H^*$ on $[0,Q^*(1)]$ which admits $Q^*$ as a best response; we are now ready to establish our saddle-point result.

\subsection{Putting it All Together}\label{sec:main-result}

With all of the ingredients assembled in the previous sections, we can now combine everything together to prove our main result: characterization of a saddle point for the problem of minimizing worst-case regret.

\begin{theorem}\label{thm:main-result}
	 For every value distribution $F \in \Delta([0,1])$ with $F(0) = 0$, there exists
	 \begin{itemize}
	 	\item a quantile-based bidding strategy $Q^*: [0,1] \to [0,1]$ which is a solution to the initial value problem
	 		\begin{align}
	 			x'(t)\ =\ \frac{F^-(t) - x(t)}{1 - F(x(t))}; \qquad x(0) = 0\; \label{eq:ODE_theorem}
	 		\end{align}
	 		and satisfies the properties outlined in \Cref{lemma:ode-existence}; and
	 	\item a distribution $H^*$ of highest competing bids defined for every $x \in [0,Q^*(1)] \subseteq [0,1]$ as
	 		\begin{align*}
				H^*(x)\ \coloneqq\ \exp\left( - \int_{(Q^*)^{-1}(x)}^1 \frac{1}{1 - F(Q^*(t))} dt \right)\,,
			\end{align*}
	 \end{itemize}
	 such that the pair $(\s[Q^*], H^*)$, where $\s[Q^*]$ is the bidding strategy corresponding to $Q^*$ (see \Cref{def:quantile-based-bidding-strategy}), satisfies,
     \begin{equation*}
         \inf_{s \in \Scal}\ \sup_{H \in \Delta([0,1])} \rr[F]{s, H} = \rr[F]{\s[Q^*], H^*} = \int_0^1 Q^*(t) dt.
     \end{equation*}
\end{theorem}

\Cref{thm:main-result} and \Cref{thm:full-info} together yield a saddle point for regret against the original partial-information benchmark, as the following corollary notes. And, in turn, this fact immediately implies the minimax optimality of the bidding strategy $\s[Q^*]$.
\begin{corollary}\label{cor:partial-info-saddle}
	The pair $(s_{Q^*}, \mathcal H^*)$, where $\mathcal{H}^*$ is the distribution of $\delta_h$ when $h \sim H^*$, satisfies
    \begin{align*}
        \inf_{s \in \Scal}\ \sup_{H \in \Delta([0,1])}\ \reg_F(s, H)\ =\ \reg_F(\s[Q^*], \mathcal H^*)\ =\   \int_0^1 Q^*(t)dt\,,
    \end{align*}
    with $\reg_F(s, \mathcal H^*) \coloneqq \E_{H \sim \mathcal{H^*}}[\reg_F(s,H)]$. In particular, $\s[Q^*]$ is a minimax-optimal bidding strategy.
\end{corollary}

\Cref{thm:main-result} establishes our main result. It characterizes the minimax-optimal bidding strategy $\s[Q^*]$ for the full-information benchmark, which we then extend to the partial-information benchmark in \Cref{thm:full-info}. Crucially, \Cref{thm:main-result} provides an efficient procedure for the construction of $Q^*$---and consequently the minimax-optimal bidding strategy $\s[Q^*]$---as a solution to the ordinary differential equation \eqref{eq:ODE_theorem}. The corresponding minimax-optimal regret can then be determined through a simple integral evaluation. This framework drastically reduces the complexity faced by the buyer: instead of contending with an intricate infinite-dimensional minimax optimization problem, she only needs to solve an ordinary differential equation. In addition to computational tractability, our framework offers a wide range of structural insights about minimax-optimal bidding in first-price auctions. For starters, our construction of $\s[Q^*]$ yields a \emph{deterministic} minimax-optimal bidding strategy for continuous value distributions. This fact follows directly from \Cref{def:quantile-based-bidding-strategy} because $F(\{v\}) = 0$ holds for all $v \in [0,1]$ when $F$ has a continuous CDF.

\begin{corollary}\label{cor:deterministic-opt-strat}
    For any value distribution $F$ with a continuous CDF, the minimax-optimal bidding strategy $\s[Q^*]$ is deterministic.
\end{corollary}
When contrasted with \Cref{example:atom-at-1}, \Cref{cor:deterministic-opt-strat} reveals an important insight: even infinitesimal variations in values can make deterministic strategies minimax-optimal. While \Cref{example:atom-at-1} demonstrates the strict sub-optimality of deterministic strategies for distributions with atoms, these can be approximated
\footnote{The notion of approximation can be made precise with the Prokhorov or Wasserstein metrics.} 
arbitrarily well by continuous distributions, all of which admit deterministic minimax-optimal strategies. Thus, any value distribution can be perturbed ever so slightly to yield another one which admits a deterministic minimax-optimal bidding strategy.

\section{Impact of the Value Distribution}\label{sec:value-dist-impact}

Beyond characterizing a minimax-optimal bidding strategy, \Cref{thm:main-result} also quantifies the buyer’s minimax regret, providing a valuable metric for assessing the impact of informational deficiency in bidding environments. In this section, we show that our result facilitates comparisons of worst-case regret across different value distributions $F$, allowing buyers to ascertain whether acquiring information about competing bids could substantially improve their performance.

We first characterize the value distribution which maximizes the minimax regret, i.e., the one for which the buyer is suffering the most from the information deficiency. 
For $\rho \in [1,\infty]$, let $\mathcal{F}_\rho$ denote the set of distributions $D$ on $[0,1]$ such that $D(A) \leq \rho \cdot \lambda(A)$ for all $A \in \borel$. Here, $\mathcal F_\infty$ is simply the set $\Delta([0,1])$ of all distributions on $[0,1]$. Intuitively, $\mathcal F_\rho$ is the set of distributions whose ``density'' is bounded above by $\rho$.
Our goal is to solve, for every $\rho \in [1,\infty]$, the following optimization problem:
\begin{equation*}
    \sup_{F \in \mathcal{F}_{\rho}}\ \inf_{s \in \Scal}\ \sup_{H \in \Delta([0,1])}\ \reg_F(s,H)\,.
\end{equation*}
We note that this max-min-max problem introduces an additional maximization layer over the infinite-dimensional space of value distributions, further complicating the already challenging minimax problem solved in \Cref{sec:minimax}. Nevertheless, our next result shows that the problem remains solvable, allowing us to characterize the worst-case value distribution for every $\rho \in [1,\infty].$ 
\begin{theorem}\label{thm:worst-value-dist}
	For every $\rho \in [1,\infty]$, we have
	\begin{align*}
		\unif(1 - \tfrac{1}{\rho}, 1)\ \in\ \argmax_{F \in \mathcal F_\rho}\ \inf_{s \in \Scal}\ \sup_{H \in \Delta([0,1])}\ \reg_F(s,H)\,,
	\end{align*}
	where $\unif(1 - \tfrac{1}{\infty},1) \coloneqq \delta_1$.
\end{theorem}
\Cref{thm:worst-value-dist} shows that, for every $\rho \in [1,\infty]$, the value distribution that maximizes the buyer’s minimax regret is the one that concentrates mass as much as possible among all distributions in $\mathcal{F}_{\rho}$. Indeed, the regret-maximizing distribution $F_{\rho} = \unif(1 - \tfrac{1}{\rho}, 1)$ satisfies $F_{\rho}(A) = \rho \cdot \lambda(A)$ for every measurable set $A$ included in the support of $F_{\rho}$. 

\Cref{fig:regret_unif} plots the regret of our minimax-optimal strategy when the value distribution follows a $\unif(a,1)$ for $a \in [0,1)$. With $\rho = 1/(1 - a)$, the worst-case regret at $a$ reflects the largest regret that is incurred by our minimax-optimal strategy across all distributions in $\mathcal F_\rho$. We see that the worst-case regret improves significantly as the value distribution becomes less concentrated. In particular, the minimax-optimal regret increases by over 100\%  when $a$ shits from $0$, representing the uniform distribution, to $1$, representing the completely concentrated $\delta_1$ distribution.
\begin{figure}[h!]
    \centering
    \begin{tikzpicture}[scale = 0.65]
    \begin{axis}[
        width=10cm,
        height=10cm,
        xmin=0,xmax=1,
        ymin=0,ymax=0.6,
        table/col sep=comma,
        xlabel={$a$},
        ylabel={Worst-case regret},
        grid=both,
        legend pos=north west,
    ]

    \addplot [blue,  line width=0.7mm, ,mark=square,mark options={scale=.1}] table[x=a,y={regret}] {Data/regrets_uniform_a.csv};
    \addlegendentry{Minimax strategy}

    \addplot [violet,  line width=0.7mm, ,mark=square,mark options={scale=.1}] table[x={rho},y={best_performance}] {Data/best_alpha_and_performance_uniform_dist.csv};
    \addlegendentry{Optimal uniform-bid-shading}
    
    \addplot [black, line width = 0.7mm] {0.25*x + 0.25};  
   \addlegendentry{$\alpha = 0.5$}
      \addlegendimage{ultra thick,black}
    \end{axis}
    \end{tikzpicture}
    \caption{\textbf{Impact of the value distribution on performance.} We report the worst-case regret of various bidding strategies as a function of $a$, when the value distributions is of the form $\unif(a,1)$. The black curve corresponds to the uniform-bid-shading strategy with $\alpha = 0.5$, the violet curve to the strategy which uses the optimal uniform-bid-shading strategy for each $a$, and the blue curve to the minimax-optimal strategy.} 
    \label{fig:regret_unif}
\end{figure}
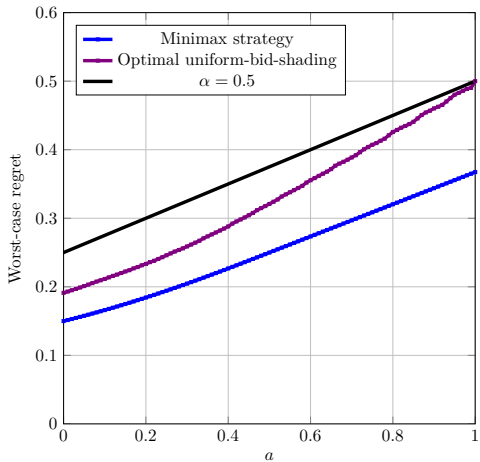

\Cref{fig:regret_unif} also quantifies the benefit one can derive from accounting for the value distribution, instead of treating each value in isolation. To see this, consider a tale of two advertisers (buyers), called Alice and Bob, who wish to launch an online ad campaign. Each of them is going to participate in a large number of first-price auctions independently of each other. Both lack information about the competition, and wish to develop a bidding strategy to minimize worst-case regret. Additionally, they have a preference for deterministic strategies due to their simplicity and reproducibility. Moreover, they share the same value distribution $F$. Although they face the same problem, they diverge in their solution approaches. Alice accounts for the fact that her values are random and distributed according to $F$, and thus elects to use our minimax-optimal strategy $\s[Q^*]$. Bob, on the other hand, ignores the value distribution and treats each value in isolation. Therefore, he uses the result of \citet{kasberger2023robust} and selects the deterministic minimax-optimal bid $b = 0.5 \cdot v$ for each value $v$, i.e., he uses the uniform-bid-shading strategy with $\alpha = 0.5$. Unlike Alice, Bob is (implicitly) making the pessimistic assumption that the competition has knowledge of his value $v$, and can customize their bids using that information. \Cref{fig:regret_unif} depicts the outcome of this tale with a performance comparison of the two strategies. The minimax-optimal strategy of Alice consistently and significantly outperforms the uniform-bid-shading strategy of Bob. For example, when the values are uniformly distributed ($a = 0$), the worst-case regret of the uniform-bid-shading strategy with $\alpha = 0.5$ is 66\% higher than the worst-case regret of the minimax-optimal strategy. The primary driver of this gap in performance is the ability of the minimax-optimal strategy to account for the value distribution $F$. In particular, the greater the variation in values, the greater the information asymmetry between the buyer and the competition. The minimax-optimal strategy exploits this asymmetry and leverages the inherent variation in values to hedge against bad outcomes. If one ignores the value distribution, as Bob did, it leads to an overly-pessimistic strategy which assumes that private information, which is not even available to the competition, can be used to hurt the buyer. This prevents it from leveraging the inherent variation in values, and leads to consistently sub-optimal performance across all distributions.

Finally, our minimax-optimal strategy also outperforms uniform-bid-shading strategies as a class. In particular, it does consistently better than the optimal uniform-bid-shading strategy which accounts for the value distribution, improving performance across all values of $a$. The improvement comes in spite of the unfavorable choice of distribution: \Cref{thm:worst-value-dist} implies that $\unif(a,1)$ leads to the largest minimax-optimal regret among all value distributions in $\mathcal{F}_{\frac{1}{1-a}}$.

\section{Conclusion}

This paper develops a prior-independent framework for designing bidding strategies in first-price auctions, using worst-case regret as the performance metric. Our approach bypasses the need for strong distributional assumptions about competing bids, making it robust to uncertainty in auction environments. We provide an efficient procedure to (i)~evaluate and compare the worst-case regret of arbitrary strategies, and (ii)~construct a minimax-optimal bidding strategy as solution to an ODE. Importantly, our method works for every value distribution; a fact we use to analyze its impact on regret. Our minimax-optimal strategy provides a principled way for buyers to contend with uncertainty about the competition, whether due to a lack of data or rapid changes in the market. Moreover, it consistently attains lower regret than the uniform-bid-shading strategies advocated by prior work.

{\setstretch{1.0}
\bibliographystyle{agsm}
\bibliography{ref}}

\newpage

\appendix

\allowdisplaybreaks
\section{Proof of Result in \Cref{sec:strat_evaluation}}

\begin{proof}[\textbf{Proof of \Cref{thm:evaluation}}]
To prove this result, we will apply Bauer Maximum Principle (see 7.69 of \citealt{aliprantis2006infinite}). 

In a first step, we establish that the mapping $H \mapsto \reg_F(s,H) = \sup_{s' \in \Scal} \U[F]{s',H} - \U[F]{s,H}$ is convex and upper semi-continuous on $\Delta([0,1])$ under the weak topology.

We start by establishing the upper semi-continuity of the utility $u(b,h;v)$ of the buyer. In particular, we claim that $(b,h, v) \mapsto u(b,h;v) = (v - b) \mathbbm{1}(b \geq h)$ is upper semi-continuous on $[0,1]^3$. To see this, note that
	\begin{itemize}
		\item $(b,h,v) \mapsto (v-b)$ is continuous and therefore upper semi-continuous on $[0,1]^3$;
		\item $(b,h,v) \mapsto \mathbbm{1}(b \geq h)$ is upper semi-continuous on $[0,1]^3$ because for every $y \in \mathbb{R}$, the level set $\{(b,h,v) \mid \mathbbm{1}(b \geq h) \geq y\}$ is either empty or equal to $\{(b,h,v) \mid b \geq h\}$, both of which are closed.
	\end{itemize} 
	As the product of two upper semi-continuous functions is upper semi-continuous, we get that $u(b,h;v)$ is upper semi-continuous on $[0,1]^3$, as desired.
	
	Next, we show that $\mathcal{U}: (J,H) \mapsto  \E_{((v,b), h) \sim J \times H}[u(b,h;v)]$ is upper semi-continuous on $\Delta([0,1]^2) \times \Delta([0,1])$ with respect to the weak topology. Let $(J,H) \in \Delta([0,1]^2) \times \Delta([0,1])$ and consider a sequence $(J_n, H_n)$ which converges to $(J,H)$ in the weak topology. As we established that $u$ upper semi-continuous on $[0,1]^3$, we have by the Portmonteau lemma that $\limsup_{n \to \infty} \E_{J_n \times H_n}[u(b,h;v)] \leq \E_{J \times H}[u(b,h;v)]$. Therefore, $U$ is upper semi-continuous on\\ ${\Delta([0,1]^2) \times \Delta([0,1])}$ with respect to the weak topology.
	
	Now, recall that each $s' \in \Scal$ has an associated joint distribution $J_{s',F}$ of value-bid pairs $(v,b)$ such that the marginal distribution of $v$ is $F$. Moreover, by Fubini's Theorem, we have
	\begin{align*}
		\U[F]{s',H} = \mathbb{E}_{(v, h) \sim F \times H}  \left[ \mathbb{E}_{b \sim s'(v)} \left[ u(b,h;v) \right] \right] = \E_{((v,b), h) \sim J_{s',F} \times H}[u(b,h;v)] = \U{J_{s',F},H}\,.
	\end{align*}
	
    For a given $F$, let $\mathcal J \coloneqq \{J \in \Delta([0,1]^2) \mid \Prob_{(v,b) \sim J}(v \leq x) = F(x)\}$ be the set of all value-bid joint distributions with $F$ as the marginal distribution of $v$. Then, for every highest-bid distribution $H \in \Delta([0,1])$, we have
	\begin{align*}
		\sup_{s' \in \Scal}\ \U[F]{s',H}\ =\ \sup_{J \in \mathcal J}\ \U{J,H}\,.
	\end{align*}
	It is straightforward to check that $\mathcal J$ is closed under the weak topology. We also note that $\Delta([0,1]^2)$ is compact under the weak topology as it is the set of distributions over a compact domain. Therefore, $\mathcal{J}$ is compact as it is a closed subset of a compact set. Combining this with the upper semi-continuity of $\U{J,H}$ allows us to invoke the Maximum Theorem (see Lemma~17.30 of \citealt{aliprantis2006infinite}) , which implies that the optimal-value function $H \mapsto \sup_{J \in \mathcal J} \U{J,H}$ is also upper semi-continuous.
	
    To conclude that $H \mapsto \reg_F(s,H)$ is upper semi-continuous, we need one more ingredient: continuity of $H \mapsto \U[F]{s,H}$. 
    Consider the function $f(h) \coloneqq \E_{v \sim F, b \sim s(v)}[(v-b) \cdot \mathbbm{1} \{b \geq h\}]$. Note that, as $f$ is bounded,  it is sufficient to prove that $f$ is continuous on $[0,1]$ to conclude that $H \mapsto \U[F]{s,H}$ is continuous in $\Delta([0,1])$ under the weak topology. This follows from the definition of the weak convergence. 
    
    Let us prove that $f$ is continuous. Recall that by assumption, the bid distribution $P_{s,F}$ is absolutely continuous. Hence, for every $\epsilon > 0$, there exists a $\delta$ such that for every measurable set $A \subset [0,1]$, we have that $P_{s,F}(A) \leq \epsilon$ whenever $\lambda(A) \leq \delta$, where $\lambda(\cdot)$ denotes the Lebesgue measure on $[0,1]$. Moreover, remark that for every $h_1,h_2 \in [0,1]$ we have that
    \begin{align*}
        |f(h_1) - f(h_2)| 
        &= \left| \E_{v \sim F, b\sim s(v)}[(v - b) \mathbbm{1}(\min\{h_1, h_2\} \leq b < \max\{h_1, h_2\}) \right| \\
        &\leq P_{s,F}(\min\{h_1, h_2\} \leq b < \max\{h_1, h_2\}) \,.
    \end{align*}
    Hence, for every $\epsilon > 0$, there exists $\delta$ such that for every for every $h_1,h_2 \in [0,1]$ with $|h_1 - h_2| \leq \delta$, we have that
	\begin{align*}
		|f(h_1) - f(h_2)| \leq
		P_{s,F}(\min\{h_1, h_2\} \leq b < \max\{h_1, h_2\})
		\leq \epsilon\,.
	\end{align*}
	This implies that $f$ is continuous. Hence, $H \mapsto \U[F]{s,H}$ is continuous.
    
	Altogether, we can combine the upper semi-continuity of $H \mapsto \sup_{s' \in \Scal} \U[F]{s',H}$ and the continuity of $H \mapsto \U[F]{s,H}$ to get the upper semi-continuity of $H \mapsto \reg_F(s,H) = \sup_{s' \in \Scal} \U[F]{s',H} - \U[F]{s,H}$. 
    
    Moreover, as $H \mapsto \reg_F(s,H)$ is the pointwise supremum of linear functionals $H \mapsto \U[F]{s',H} - \U[F]{s,H}$ on $\Delta([0,1])$, it is convex. Together with the fact that $\Delta([0,1])$ is a compact and convex locally-convex Hausdorff space in the weak topology, this allows us to apply the Bauer Maximum Principle (see 7.69 of \citealt{aliprantis2006infinite}). Therefore,  $H \mapsto \reg_F(s,H)$ has a maximizer that is an extreme point of $\Delta([0,1])$, i.e., it has a maximizer of the form $H = \delta_h$ for some $h \in [0,1]$. Therefore, we have shown the first equality of the proposition, namely
	\begin{align*}
		\sup_{H \in \Delta([0,1])} \reg_F(s,H)\ =\ \sup_{h\in [0,1]} \reg_F(s,\delta_h)\,.
	\end{align*}
	To see the second equality, observe that if the highest competing bid distribution is known to be the point mass on $h$, i.e., $H = \delta_h$, then the utility-maximizing bidding strategy simply bids $h$ whenever the value $v$ is greater than or equal to $h$. In other words, $\reg_F(s,\delta_h)\ =\ \E_{v \sim F}\left[ (v-h) \cdot \mathbbm{1}(v \geq h) \right]\ -\ \U[F]{s,\delta_h}.$
\end{proof}

\section{Proofs of Results in \Cref{sec:minimax}}

\subsection{Proof of \Cref{remark:kernel}}\label{appendix:minimax-remark}
Given a quantile-based bidding strategy $Q \in \mathcal{Q}$, we note that for any fixed $B = (b,1] \in \borel$, the map 
	\begin{align*}
		v \mapsto \kappa_{\s[Q]}(B, v) = \begin{cases}
			\mathbbm{1} \{ Q(F(v)) \in B \} &\text{if }Y(v) \in \{\{F(v)\}, \emptyset\}\\
			\frac{\lambda(Q^{-1}(B) \cap Y(v))}{\lambda(Y(v))} &\text{otherwise}
		\end{cases}
	\end{align*} 
is a non-decreasing function, and therefore measurable, because $Y(v_1)\succcurlyeq Y(v_2)$ for all $v_1 \geq v_2$. As ${\{(b,1]\mid b\in [0,1]\}}$ generates $\mathcal B$, $\kappa_{\s[Q]}$ is a Markov kernel and $\s[Q] \in \Scal$.

\subsection{Proof of Lemmas in \Cref{sec:minimax}}\label{appendix:minimax-lemmas}

\begin{proof}[\textbf{Proof of \Cref{lemma:quantile-based-bidding}}]
	Define $V \coloneqq \{v \in[0,1] \mid \lambda(Y(v)) = F(\{v\}) > 0\}$ to be the set of values $v$ for which $\s[Q]$ plays a randomized bid. Since every distribution admits at most countably many atoms, the set $V$ must be countable.
	
	For every $v \in V$, we have
	\begin{align*}
		\E_{b \sim \s[Q](v)} [(v - b) \mathbbm{1}(b \geq h)] = \E_{t \sim \unif(Y(v))}[(v - Q(t))\mathbbm{1}(Q(t) \geq Q(y))] = \E_{t \sim \unif(Y(v))}[(v - Q(t))\mathbbm{1}(t \geq y)]
	\end{align*}
	Now, for any interval $Y(v)$ with $\lambda(Y(v))> 0$, the uniform distribution on $Y(v)$ is the same as the uniform distribution on $[0,1]$ conditioned on the event $Y(v)$. Therefore,
	\begin{align*}
		\E_{b \sim \s[Q](v)} [(v - b) \mathbbm{1}(b \geq h)] &= \E_{t \sim \unif(0,1)}[(v - Q(t))\mathbbm{1}(t \geq y) \mid t \in Y(v)]\\
		&= \frac{1}{\lambda(Y(v))} \cdot \E_{t \sim \unif(0,1)}[(v - Q(t))\mathbbm{1}(t \geq y,  t \in Y(v))]\\
		&= \frac{1}{F(\{v\})} \cdot \E_{t \sim \unif(0,1)}[(F^-(t) - Q(t))\mathbbm{1}(t \geq y,  F^-(t) = v)]
	\end{align*}
	Taking an expectation over all $v \in V$ yields
	\begin{align*}
		\E_{v \sim F}\left[ \E_{b \sim \s[Q](v)} [(v - b) \mathbbm{1}(b \geq h)] \mathbbm{1}(v \in V) \right] &= \sum_{v \in V} F(\{v\}) \cdot \E_{b \sim \s[Q](v)} [(v - b) \mathbbm{1}(b \geq h)]\\
		&= \sum_{v \in V} \E_{t \sim \unif(0,1)}[(F^-(t) - Q(t))\mathbbm{1}(t \geq y,  F^-(t) = v)]\\
		&= \E_{t \sim \unif(0,1)}[(F^-(t) - Q(t))\mathbbm{1}(t \geq y,  F^-(t) \in V)] \tag{I}
	\end{align*}
	
	On the other hand, if we take an expectation over $v \notin V$, we get 
	\begin{align*}
		\E_{v \sim F}\left[ \E_{b \sim \s[Q](v)} [(v - b) \mathbbm{1}(b \geq h)] \mathbbm{1}(v \notin V) \right] &= \E_{v \sim F}\left[(v - Q(F(v))) \mathbbm{1}(Q(F(v)) \geq Q(y)) \mathbbm{1}(v \notin V) \right]\\
			&= \E_{v \sim F}\left[\underbrace{(v - Q(F(v))) \mathbbm{1}(Q(F(v)) \geq Q(y), v \notin V)}_{g(v)} \right]
	\end{align*}
	Recall that, if $y \sim \unif(0,1)$, then $F^-(y) \sim F$~\citep{embrechts2013note}. Therefore, $\E_{v \sim F}[g(v)] = \E_{t \sim \unif(0,1)}[g(F^-(t))]$. Now, consider $t \in [0,1]$ such that $F^-(t) \notin V$. As $Y(F^-(t))$ is an interval with measure zero and $t \in Y(F^-(t))$, we must have that $Y(F^-(t)) = \{t\}$. We claim that $F(F^-(t)) = t$. For contradiction, suppose not. Then, as $F(F^-(t)) \geq t$ always holds~\citet{embrechts2013note}, we must have $F(F^-(t)) > t$. Applying $F^-$ to both sides yields $F^-(F(F^-(t))) \geq F^-(t)$. However, we always have $F^-(F(x)) \leq x$~\citep{embrechts2013note}, which for $x = F^-(t)$ implies $F^-(F(F^-(t))) \leq F^-(t)$. Hence, we must have $F^-(F(F^-(t))) = F^-(t)$, i.e, $F^-(F(F^-(t))) \in Y(F^-(t))$. Since we assumed $F(F^-(t)) > t$, this implies $\lambda(Y(v)) > 0$, which contradicts $F^-(t) \notin V$. Thus, we must have $F(F^-(t)) = t$, and as a consequence $Q(F(F^-(t)) = Q(t)$. Altogether, we get
	\begin{align*}
		&\E_{v \sim F}\left[ \E_{b \sim \s[Q](v)} [(v - b) \mathbbm{1}(b \geq h)] \mathbbm{1}(v \notin V) \right]\\ = &\E_{v \sim F}[g(v)]\\
		= &\E_{t \sim \unif(0,1)}[g(F^-(t))]\\
		= &\E_{t \sim \unif(0,1)}\left[(F^-(t) - Q(t)) \mathbbm{1}(Q(t) \geq Q(y), F^-(t) \notin V) \right]\\
		= &\E_{t \sim \unif(0,1)}\left[(F^-(t) - Q(t)) \mathbbm{1}(t \geq y, F^-(t) \notin V) \right] \tag{II}
	\end{align*} 
	Finally, adding together (I) and (II) yields
	\begin{align*}
		& \E_{v \sim F}[ \E_{b \sim \s[Q](v)} [(v - b) \mathbbm{1}(b \geq h)]]\\
		=\ &\E_{v \sim F}\left[ \E_{b \sim \s[Q](v)} [(v - b) \mathbbm{1}(b \geq h)] \mathbbm{1}(v \in V) \right] + \E_{v \sim F}\left[ \E_{b \sim \s[Q](v)} [(v - b) \mathbbm{1}(b \geq h)] \mathbbm{1}(v \notin V) \right]\\
		=\ &\E_{t \sim \unif(0,1)}[(F^-(t) - Q(t))\mathbbm{1}(t \geq y,  F^-(t) \in V)] + \E_{t \sim \unif(0,1)}\left[(F^-(t) - Q(t)) \mathbbm{1}(t \geq y, F^-(t) \notin V) \right]\\
		=\ &\E_{t \sim \unif(0,1)}[(F^-(t) - Q(t))\mathbbm{1}(t \geq y)]\\
		=\ &\int_y^1 (F^-(t) - Q(t))dt\,.
	\end{align*}

    Similarly, for $h = Q(y)$, if we replace the random variable $(v - b) \cdot \mathbbm{1}(b \geq h)$ with $\mathbbm{1}(b \geq h)$, we get
    \begin{align*}
        \Prob(b \geq h\mid b \sim P_{\s[Q],F}) = \E_{v \sim F}[ \E_{b \sim \s[Q](v)} [\mathbbm{1}(b \geq h)]]\ =\ \E_{t \sim \unif(0,1)}[\mathbbm{1}(t \geq y)] = \Prob_{t \sim \unif(0,1)}(Q(t) \geq h) \,.
    \end{align*}
    On the other hand, if $h < Q(0)$ (respectively $h > Q(1)$), then $\Prob(b \geq h\mid b \sim P_{\s[Q],F}) = 1$ (respectively $\Prob(b \geq h\mid b \sim P_{\s[Q],F}) = 0$). Therefore, for all $h \in [0,1]$, we have
    \begin{align*}
        \Prob(b \geq h\mid b \sim P_{\s[Q],F}) = \Prob_{t \sim \unif(0,1)}(Q(t) \geq h)\,,
    \end{align*}
    Hence, the induced bid distribution $P_{s,f}$ is equal to $\{Q(t)\mid t \sim \unif(0,1)\}$.
\end{proof}

\begin{proof}[\textbf{Proof of \Cref{lemma:ode-existence}}]
	We start by rewriting the ODE \eqref{eq:first-ode} in a more standard form 
	\begin{align}\label{eq:first-ode-rewrite}
		x'(t) = g(t,x(t)) \quad \text{where}\quad g(t,x) \coloneqq \frac{F^-(t) - x}{1 - F(x)}\,,
	\end{align}
	where the denominator is only well-defined for $x < F^-(1)$. For the remainder of the proof, set $\bar v = F^-(1)$. Then, $F(x) < 1$ for all $x < \bar v$. Moreover, we assume $F(x) = 0$ for $x < 0$.
	
	To prove the existence of $Q^*$, we will use a generalization of Caratheodary's existence theorem proven in \citet{biles1997caratheodory} (see \citealt{biles2000solvability} for the more general version which applies to initial value problems). We refer to it as the Generalized Caratheodary's Existence Theorem, or simply GCET.
	\begin{theorem*}[GCET]
		Consider the initial value problem (IVP) defined by $x'(t) = f(t,x(t))$ for all $t\in [0,1]$ and $x(0) = 0$. Suppose $f:[0,1] \times \R \to \R$ satisfies the following conditions:
		\begin{itemize}
			\item[(a)] For almost all $t$, $f(t,\cdot)$ is quasi-semicontinuous, i.e., for all $x \in \R$, we have
						\begin{align*}
							\limsup_{\tilde x \uparrow x} f(t,\tilde x) \leq f(t,x) \leq \liminf_{\tilde x \downarrow x} f(t, \tilde x)\,.
						\end{align*}
			\item[(b)] For each $x \in \R$, $f(\cdot, x)$ is measurable.
			\item[(c)] There exists an integrable function $\beta:[0,1] \to \R$ such that $|f(t,x)| \leq \beta(t)$ for all $t,x$.
		\end{itemize}
		Then, there exists an absolutely continuous function $x:[0,1] \to \R$ such that $x(0) = 0$ and $x'(t) = f(t,x(t))$ almost surely on $[0,1]$. We call such a function a \emph{solution} of the initial value problem (IVP).
	\end{theorem*}
	
    We cannot directly apply GCET to $g$ because it cannot be bounded with an integrable function; in fact $g$ is not well-defined for all $x \in \R$. This motivates us to consider a modified IVP parameterized by $\alpha > 0$: $x'_\alpha(t) = g_\alpha(t, x_\alpha(t))$ for all $t \in [0,1]$ and $x_\alpha(0)= 0$, where
	\begin{align}\label{eq:IVP-alpha}
		g_\alpha(t,x) \coloneqq \frac{F^-(t) - I(x)}{1 - F(\min\{x, \alpha\})} \quad \text{with} \quad I(x) = \max\{0, \min\{x,1\}\} \,.
	\end{align}
	
	Note that, for every $\alpha \in (0,\bar v)$, the conditions of GCET are satisfied by $g_\alpha$:
	\begin{itemize}
		\item[(a)] For every $t \in [0,1]$, we have $ \limsup_{\tilde x \uparrow x} g_\alpha(t,\tilde x) \leq g_\alpha(t,x)$ as $F$ is non-decreasing and $ \liminf_{\tilde x \downarrow x} g_\alpha(t,\tilde x) = g_\alpha(t,x)$ due to the right-continuity of $F$.
		\item[(b)] For each $x \in \R$, $g_\alpha(\cdot,x)$ is measurable because $F^-$ is measurable.
		\item[(c)] We have $g_\alpha(t,x) \leq 1/(1 - F(\alpha))$ for all $t,x$.
	\end{itemize}
	
	Therefore, for every $\alpha > 0$, there exists a solution $x_\alpha$ for IVP \eqref{eq:IVP-alpha}. To complete the proof, we will show the existence of a positive constant $\alpha > 0$ such that $0 \leq x_{\alpha}(t) \leq \alpha$ for all $t \in [0,1]$, thereby making it a solution to the original IVP defined by $g$. Looking forward, we first specify the value of $\alpha$ which will make the proof work, and then focus on establishing that $0 \leq x_{\alpha}(t) \leq \alpha$ for all $t \in [0,1]$ in the remainder of the proof. 
	
Consider the auxiliary function $h:(0,\bar v] \to [0,1]$ defined as follows. For $x \in (0,\bar v]$, we let $h(x)$ be the unique $z \in [0,x]$ which satisfies 
	\begin{align*}
		\int_{0}^{z} y \cdot (1 - F(y))dy\ = \int_{z}^{x}  (1 - F(y))dy \,.
	\end{align*}
	Observe that $\int_{0}^{z} y \cdot (1 - F(y))dy$ is a continuous strictly increasing function of $z$ and $\int_{z}^{x}  (1 - F(y))dy$ a continuous strictly decreasing $z$; this is because $1 - F(y) > 0$ for all $y <\bar v$. Thus, the Intermediate Value Theorem ensures that $h$ is well-defined. 
    
    Moreover, we have that $h(x) < x$ for all $x \in (0,\bar v]$ because $\int_{x}^{x}  (1 - F(y))dy = 0$ but $\int_{0}^{x} y \cdot (1 - F(y))dy > 0$, and we have that $h(x) > 0$ for all $x \in (0,\bar v]$ because $\int_{0}^{0} y \cdot (1 - F(y))dy = 0$ but $\int_{0}^{x}  (1 - F(y))dy > 0$.
    
    We define our candidate $\alpha$ as, $\alpha \coloneqq (h(\bar v) + \bar v)/2$ and note that, $0 < h(\bar v) < \alpha < \bar v$. By definition of $\bar v$ we have that $F(\alpha) < 1$.
	
	In order to bound $x_\alpha(t)$ from above, we will compare it with solutions of the parameterized family of IVP's defined for every $n \geq 1$, by $  \bar x'_n(t) = \bar g_n(t,x(t))$ with $\bar x_n(0) =0$, where
	\begin{align}\label{eq:IVP-n}
		\bar g_n(t,x)\ \coloneqq\ \frac{\{F^-(t) + 1/n\} \cdot  (1 - x)}{1 - F(x)} \,.
	\end{align}
	Note that $\bar g_n(t,x) \geq g_\alpha(t,x)$, and therefore, intuitively, we should have $\bar x_n(t) \geq x_\alpha(t)$; we formally show this fact now.
	
	First, we explicitly construct a solution for the IVP corresponding to $\bar g_n(t,x)$. Consider $\gamma: [0,\alpha] \to \R_+$ defined for every $x \in [0,\alpha]$ as 
	\begin{align*}
		\gamma(x)\ \coloneqq\ \int_{0}^x \frac{1 - F(z)}{1 - z} dz\,.
	\end{align*}  
    Note that $\gamma$ is a strictly increasing function because $1 - F(z) > 0$ for all $z < \bar v$. Therefore, $\gamma$ is invertible on its range which contains $[0, \bar \gamma]$, where  $\bar \gamma\ \coloneqq\ \E_F[v]\ +\ \frac{1 - F(\alpha)}{1 - h(\bar v)} \cdot \frac{\bar v - h(\bar v)}{2}\,.$
	To see this, observe that
	\begin{align*}
		\gamma(\alpha) &= \int_{0}^{h(\bar v)} \frac{1 - F(z)}{1 - z} dz\ +\ \int_{h(\bar v)}^\alpha \frac{1 - F(z)}{1 - z} dz\\
		&\geq \int_{0}^{h(\bar v)} (1 - F(z))(1 + z) dz\ +\ \frac{1 - F(\alpha)}{1 - h(\bar v)} \cdot (\alpha - h(\bar v))\\
		&= \int_{0}^{h(\bar v)} (1 - F(z)) dz\ +\ \int_{0}^{h(\bar v)} z \cdot (1 - F(z)) dz\ +\ \frac{1 - F(\alpha)}{1 - h(\bar v)} \cdot (\alpha - h(\bar v))\\
		&= \int_{0}^{h(\bar v)} (1 - F(z)) dz\ +\ \int_{h(\bar v)}^{1} (1 - F(z)) dz\ +\ \frac{1 - F(\alpha)}{1 - h(\bar v)} \cdot (\alpha - h(\bar v))\\
		&= \E_F[v]\ +\ \frac{1 - F(\alpha)}{1 - h(\bar v)} \cdot (\alpha - h(\bar v))\,.
	\end{align*}

	Let $n$ be such that $\E_F[v] + 1/n \leq \bar \gamma$ and define $\bar x_n: [0,1] \to [0,\alpha]$ as follows
	\begin{align*}
		\bar x_n(t)\ \coloneqq\ \gamma^{-1}\left( \int_0^t \{F^-(z) + 1/n\} dz \right)\,.
	\end{align*}
	$\bar x_n$ is well-defined because $\int_0^t \{F^-(z) + 1/n\} dz \leq \int_0^1 F^-(z) dz + 1/n = \E_F[v] + 1/n$. Since $\bar x_n$ is a composition of two strictly increasing functions, it is itself strictly increasing. Moreover, we also have $\bar x_n(0) =0$.
	
	Let $\Gamma_n$ be the set of points $t \in (0,1)$ such that $F^-$ is continuous at $t$ and $F$ is continuous at $\bar x_n(t)$. Since both $F^-$ and $F$ are discontinuous at atmost countably many points, and $\bar x_n$ is a strictly increasing function, we get that $\Gamma_n$ has measure 1, i.e., $\lambda(\Gamma_n) = 1$. Fix a point $t \in \Gamma_n$. Then, as $F$ is continuous at $\bar x_n(t)$, the Fundamental Theorem of Calculus implies that $\gamma$ is differentiable at $\bar x_n(t)$ with $\gamma'(\bar x_n(t)) = (1- F(\bar x_n(t)))/(1 - \bar x_n(t)) > 0$. Consequently, we get that $\gamma^{-1}$ is differentiable at $\bar x_n(t)$ and
	\begin{align*}
		(\gamma^{-1})'(\gamma(\bar x_n(t)))\ =\ \frac{1 - \bar x_n(t)}{1 - F(x_n(t))}\,.
	\end{align*}
	Finally, using the continuity of $F^-$ at $t$, the Fundamental Theorem of Calculus, and the Chain Rule for derivatives, we get
	\begin{align*}
		x'_n(t)\ =\ \frac{1 - \bar x_n(t)}{1 - F(x_n(t))} \cdot \{F^-(t) + 1/n\} &&\forall\ t\in \Gamma_n\,.
	\end{align*}
	Therefore, we have established that $\bar x_n$ is a solution of the IVP described in \eqref{eq:IVP-n}. 
	
	We are now ready to prove $\bar x_n(t) \geq x_\alpha(t)$ for all $t \in [0,1]$. For contradiction, suppose not and let $t^* = \inf\{t \in [0,1] \mid x_\alpha(t) > \bar x_n(t)\} \in [0,1)$. We note that by continuity of the functions $x_\alpha$ and $\bar x_n$ and because $x_\alpha(0) = \bar x_n(0) =0$ we have that $x_\alpha(t^*) = \bar x_n(t^*) = x^*$. 
    
    Consider a sequence $t_k \downarrow t^*$ such that, for all $k \geq 1$, we have
	\begin{align*}
		\frac{x_\alpha(t_k) - x_\alpha(t^*)}{t_k - t^*}\ >\ \frac{\bar x_n(t_k) - \bar x_n(t^*)}{t_k - t^*}\, \tag{\#}.
	\end{align*}
	As $x_\alpha$ is a solution of IVP \eqref{eq:IVP-alpha}, we get
	\begin{align*}
		\frac{x_\alpha(t_k) - x_\alpha(t^*)}{t_k - t^*} &= \frac{1}{t_k - t^*} \cdot \int_{t^*}^{t_k} g_\alpha(z, x_\alpha(z)) dz\\ 
		&= \frac{1}{t_k - t^*} \cdot \int_{t^*}^{t_k} \frac{F^-(z) - I(x_\alpha(z))}{1 - F(\max\{x_\alpha(z), \alpha\})} dz\\
		&\leq \frac{1}{t_k - t^*} \cdot \int_{t^*}^{t_k} \frac{F^-(t_k) - I(x_\alpha(t^*))}{1 - F(\max\{x_\alpha(t_k), \alpha\})} dz\\
		&= \frac{F^-(t_k) - x^*}{1 - F(\max\{x_\alpha(t_k), \alpha\})}\,.
	\end{align*}
	Similarly, as $\bar x_n$ is a solution of IVP \eqref{eq:IVP-n}, we get
	\begin{align*}
		\frac{\bar x_n(t_k) - \bar x_n(t^*)}{t_k - t^*} &= \frac{1}{t_k - t^*} \cdot \int_{t^*}^{t_k} \bar g_n(z, \bar x_n(z)) dz\\
		&= \frac{1}{t_k - t^*} \cdot \int_{t^*}^{t_k} \frac{\{F^-(t) + 1/n\} \cdot  (1 - \bar x_n(z))}{1 - F(\bar x_n(z))} dz\\
		&\geq \frac{1}{t_k - t^*} \cdot \int_{t^*}^{t_k} \frac{\{F^-(t^*)_+ + 1/n\} \cdot  (1 - \bar x_n(t_k))}{1 - F(\bar x_n(t^*))} dz\\
		&= \frac{\{F^-(t^*)_+ + 1/n\} \cdot  (1 - \bar x_n(t_k))}{1 - F(x^*)}\,,
	\end{align*}
	where $F^-(t^*)_+$ is the right limit of $F^-$ at $t^*$, i.e., $F^-(t^*)_+ = \lim_{s \downarrow t^*} F^-(s)$. Combining this with $(\#)$ yields
	\begin{align*}
		\frac{F^-(t_k) - x^*}{1 - F(\max\{x_\alpha(t_k), \alpha\})} \geq \frac{\{F^-(t^*)_+ + 1/n\} \cdot  (1 - \bar x_n(t_k))}{1 - F(x^*)} &&\forall\ k \geq 1\,.
	\end{align*}
	Note that $x_\alpha$ and $\bar x_n$ are both continuous and non-decreasing, and $F$ is right continuous. Therefore, we have that $\lim_{t_k \downarrow t^*} F(x_\alpha(t_k)) = F(x^*)$ and $\lim_{t_k \downarrow t^*} \bar x_n(t_k) = x^*$. Hence, taking the limit as $t_k \downarrow t^*$ yields
	\begin{align*}
		\frac{F^-(t^*)_+ - x^*}{1 - F(\max\{x^*, \alpha\})} \geq \frac{\{F^-(t^*)_+ + 1/n\} \cdot  (1 -  x^*)}{1 - F(x^*)}\,.
	\end{align*}
	This gives the desired contradiction because $1 - F(\max\{x^*, \alpha\}) \geq 1 - F(x^*)$ and
	\begin{align*}
		&F^-(t^*)_+ - x^* \leq F^-(t^*)_+ \cdot (1 - x^*) <  \{F^-(t^*)_+ + 1/n\} \cdot  (1 -  x^*)\,.
	\end{align*}
    where we have used the fact that $x^* = \bar x_n(t) \leq \alpha < 1$. Hence, we must have $x_\alpha(t) \leq \bar x_n(t)$ for all $t \in [0,1]$. Finally, sending $n \to \infty$ and using the continuity of $\gamma^{-1}$ yields for every $t \in [0,1]$,
	\begin{align*}
		x_\alpha(t)\ \leq\ \lim_{n \to \infty}\ \bar x_n(t)\ =\ \bar x(t) \coloneqq \gamma^{-1}\left( \int_0^t F^-(z) dz \right)\, \tag{$\spadesuit$}
	\end{align*}

    We are now ready to finish the proof of the lemma. In particular, note that for all $t \in [0,1]$, we can use the monotonicity of $\gamma$ in combination with $(\spadesuit)$ to get
	\begin{align*}
		\gamma(\max\{x_\alpha(t), 0\})\ &\leq \int_0^t F^-(z) dz\\ 
		&\leq \int_0^{F^-(t)} (1 - F(z))dz\\ 
		&= \int_0^{h(F^-(t))} (1 - F(z))dz + \int_{h(F^-(t))}^{F^-(t)} (1 - F(z))dz\\
		&= \int_0^{h(F^-(t))} (1 - F(z))dz + \int_0^{h(F^-(t))} (1 - F(z))\cdot z dz\\
		&= \int_0^{h(F^-(t))} (1 - F(z))(1 + z)dz\\
		&\leq \int_0^{h(F^-(t))} \frac{1 - F(z)}{1 - z}dz\\
		&= \gamma(h(F^-(t)))\,.
	\end{align*}
	As $h(x) < x$ for all $x \in (0,\bar v]$ and $\gamma$ is strictly increasing, we get $x_\alpha(t) \leq h(F^-(t)) < F^-(t)$ for all $t \in (0,1]$. Therefore, $g_\alpha(t, x_\alpha(t)) \geq 0$ for all $t \in [0,1]$, and consequently $x_\alpha(t) \geq 0$ for all $t \in [0,1]$. Hence, there exists a set $A\subset [0,1]$ with measure $\lambda(A) = 1$ such that, for all $t \in [0,1]$, we have
	\begin{align*}
		x_\alpha'(t)\ =\ \frac{F^-(t) - I(x_\alpha(t))}{1 - F(\min\{x_\alpha(t), \alpha\})} =\ \frac{F^-(t) - x_\alpha(t)}{1 - F(x_\alpha(t))}\,,
	\end{align*}
	where we have used the fact that $x_\alpha(t) \leq h(\bar v) \leq \alpha$ for all $t \in [0,1]$. In other words, $Q^* = x_\alpha$ satisfies part 1 of the lemma.
	
	For parts $2$ and $3$, note that $x_\alpha$ being a solution of IVP \eqref{eq:IVP-alpha} implies
	\begin{align*}
		x_\alpha(t)\ =\ \int_0^t g_\alpha(z, x_\alpha(z)) dz\ =\ \int_0^t \frac{F^-(z) - x_\alpha(z)}{1 - F(x_\alpha(z))} dz\,.
	\end{align*}
	As $F^-(z) > 0$ and $x_\alpha(z) < F^-(z)$ for all $z > 0$, we get that $x_\alpha(t)$ is strictly increasing as a function of $t$, and $x_\alpha(t) > 0$ for all $t >0$. Since we have already established $x_\alpha(t) < F^-(t)$ for all $t \in (0,1]$, we get that $Q^* = x_\alpha$ also satisfies parts 2 and 3 of the lemma.
\end{proof}

\begin{proof}[\textbf{Proof of \Cref{lemma:saddle-max-prob}}]

We would like to show that the function
\begin{equation*}
		y\ \mapsto\ \int_{Q^*(y)}^1 (1 - F(t)) dt\ -\ \int_y^1 (F^-(t) - Q^*(t)) dt,
\end{equation*}
is constant on $ [0,1]$. Part~1 of \Cref{lemma:ode-existence} established that this mapping has a vanishing derivative at any point where it is differentiable. To conclude that it is constant, we will show that it is absolutely continuous using the following lemma
\begin{lemma}\label{lemma:regret-abs-cont}
	If $Q:[0,1] \to [0,1]$ is absolutely continuous, then the mapping
	\begin{align*}
		y\ \mapsto\ \int_{Q(y)}^1 (1 - F(t)) dt\ -\ \int_y^1 (F^-(t) - Q(t)) dt
	\end{align*}
	 is also absolutely continuous on $[0,1]$.
\end{lemma}
\begin{proof}
	First, note that the second term is absolutely continuous on $[0,1]$ because it is the integral of a bounded integrable function $t \mapsto F^-(t) - Q(t)$. Next, observe that the function $g(x) = \int_x^1 (1 - F(t))dt$ is Lipschitz continuous with Lipschitz constant 1 because for $x < z$
	\begin{align*}
		|f(x) - f(z)| \leq \int_x^z (1 - F(t)) dt \leq z - x\,. 
	\end{align*}
	Since the composition of an absolutely continuous with a globally Lipschitz-continuous function is absolutely continuous, we get that
	\begin{align*}
		y \mapsto\ g\left( Q(y) \right) = \int_{Q(y)}^1 (1 - F(t)) dt
	\end{align*}
	is absolutely continuous on $[0,1]$. Finally, the difference of two absolutely continuous functions is also absolutely continuous, thereby establishing the lemma.
\end{proof}
Lemma~\ref{lemma:regret-abs-cont} in combination with Lemma~\ref{lemma:ode-existence} allows us to show that all $y \in [0,1]$ are optimal for the inner-maximization problem in the saddle-point problem \eqref{eq:simpler-saddle}, i.e. that the regret incurred by the quantile-based bidding strategy $Q^*$ is constant on $[0,1]$.

    Define $g: [0, 1] \to \R_+$ as
	\begin{align*}
		g(y)\ \coloneqq\ \int_{Q^*(y)}^1 (1 - F(t)) dt\ -\ \int_y^1 (F^-(t) - Q^*(t)) dt\,.
	\end{align*}
	Then, since $Q^*$ is absolutely continuous, Lemma~\ref{lemma:regret-abs-cont} implies that $g$ is also absolutely continuous. Therefore, $g$ is differentiable almost surely and we can alternatively write $g$ as
	\begin{align*}
		g(y)\ =\ g(0) + \int_0^y g'(t)dt.
	\end{align*}
	In order to show that $g$ is constant, we show that $g'(y) = 0$ almost surely. Let $\mathcal{Y}$ be the set of $y \in [0,1]$ such that the derivate ${Q^*}'(y)$ exists, $F$ is continuous at $Q^*(y)$, and $F^-$ is continuous at $y$. Note that $Y$ has measure $\lambda(Y) = 1$ because $Q^*$ is differentiable almost surely, and both $F$ and $F^{-}$ have only countably many discontinuities (as they are non-decreasing). For every $y \in \mathcal{Y}$, we obtain by the Chain Rule that
	\begin{align*}
		g'(y)\ &= -(1 - F(Q^*(y))\cdot {Q^*}'(y)\ + (F^-(y) - Q^*(y))\\
		&= -(1 - F(Q^*(y))\cdot \frac{F^-(y) - Q^*(y)}{1 - F(Q^*(y))}\ + (F^-(y) - Q^*(y))\\
		&= 0\,,
	\end{align*}
	where we use part 1 of Lemma~\ref{lemma:ode-existence} in the second equality. As $\mathcal{Y}$ has measure $1$, we have shown that $g'(y) = 0$ almost surely. Therefore,
	\begin{align*}
		g(y)\ =\ g(0) + \int_0^y g'(t)dt\ =\ g(0)\ =\ \int_0^1 (1 - F(t)) dt\ -\ \int_0^1 (F^-(t) - Q^*(t)) dt = \int_0^1 Q^*(t) dt\,,
	\end{align*} 
	where the last equality follows from $\int_0^1 (1 - F(t))dt = \E_F[v] = \int_0^1 F^-(t) dt$.
\end{proof}

\begin{proof}[\textbf{Proof of \Cref{lemma:second-ode}}]
    We established in part 3 of \Cref{lemma:ode-existence}  that $Q^*(1) < F^-(1)$ which implies that $F(Q^*(1)) < 1$. Note that $G:[0,1] \to [0,1]$ is defined to be
	\begin{align*}
		G(y)\ \coloneqq\ H( Q^*(y))\ =\ \exp\left(- \int_{y}^1 \frac{1}{1 - F(Q^*(t))} dt \right)\,.
	\end{align*}
	First, note that:
	\begin{itemize}
		\item $G$ is the composition of the exponential function, which is globally Lipschitz on $[-(1 - F(Q^*(1))^{-1}, 0]$;
		\item the map $y \mapsto \int_{y}^1 \frac{1}{1 - F(Q^*(t))} dt$ which is also Lipschitz because
			\begin{align*}
				\left| \int_{y}^1 \frac{1}{1 - F(Q^*(t))} dt - \int_{x}^1 \frac{1}{1 - F(Q^*(t))} dt \right| \leq \frac{1}{1 - F(Q^*(1))} \cdot |x - y|\,.
			\end{align*}
	\end{itemize}
	 Therefore, $G$ is Lipschitz continuous on $[0,1]$. In particular, it is differentiable almost surely. 
	
	Let $\Gamma$ be the set of $y \in (0,1)$ for which $y \mapsto F(Q^*(y))$ is continuous and $Q^*$ is differentiable. $\Gamma$ has measure $\lambda(\Gamma) = 1$ because $F$ is discontinuous on at most countably many points, and $Q^*$ is strictly increasing and differentiable almost surely. For every $y \in \Gamma$, the Chain Rule of derivatives yields
	\begin{align*}
		G'(y)\ =\ G(y) \cdot \frac{1}{1 - F(Q^*(y))}\ =\ G(y) \cdot \frac{(Q^*)'(y)}{F^-(y) - Q^*(y)}\,,
	\end{align*}
	where the last equality follows from part~1 of \Cref{lemma:ode-existence}. Importantly, $F(y) - Q^*(y) > 0$ for all $y > 0$ by part~3 of \Cref{lemma:ode-existence}. 
\end{proof}

\begin{proof}[\textbf{Proof of \Cref{lemma:bid-strat-optimality}}]
	Fix a value $v \in [0,1]$ and let $Y(v) =\{z \in [0,1] \mid F^-(z) = v\}$. If $Y(v)$ is empty the result is straightforward. We next assume that $Y(v)$ is not empty.
	
    We first show that,
    \begin{equation}
    \label{eq:support_reduction}
        \max_{b \in [0,1]}\ u(b \mid v, H^*)\ =\ \max_{b \in [0,Q^*(1)]}\ u(b \mid v, H^*)\,.
    \end{equation}

    Let $b > Q^*(1)$. If $b \geq v$, we note that $u(b \mid v, H^*) \leq 0 = u(0 \mid v, H^*).$ If $b \in (Q^*(1),v)$, then
	\begin{align*}
		u(b \mid v, H)\ =\ (v - b) \cdot H^*(b) \leq v - b < v - Q^*(1) \stackrel{(a)}{=} (v - Q^*(1))\cdot H^*(Q^*(1))\ =\ u(Q^*(1)\mid v,H^*)\,,
	\end{align*}
    where $(a)$ holds because $H^*(Q^*(1)) = 1$ by definition of $H^*$. Therefore $\eqref{eq:support_reduction}$ holds.
    
	As $Q^*:[0,1] \to [0, Q^*(1)]$ is a strictly increasing invertible function, we further get
	\begin{align*}
		\max_{b \in [0,1]}\ u(b \mid v, H^*)\ =\ \max_{y \in [0,1]}\ u(Q^*(y) \mid v, H^*)\,.
	\end{align*}
	Hence, to prove the lemma, it suffices to show that
	\begin{align}\label{eq:lemma-bid-strat-inter}
		Y(v)\ \subseteq\ \argmax_{y \in [0,1]}\ u(Q^*(y) \mid v, H^*)\,.
	\end{align}
	As $v$ and $H^*$ have been fixed, we set $u(y) \coloneqq u(Q^*(y) \mid v, H^*)$ to simplify notation. Observe that
	\begin{align*}
		u(y)\ =\ (v - Q^*(y)) \cdot H^*(Q^*(y))\ =\ (v - Q^*(y)) \cdot G(y)\,.
	\end{align*}
	Using the fact that the product of two absolutely continuous functions (on a bounded interval) is also absolutely continuous, we get that $y \mapsto u(y \mid v, H^*)$ is absolutely continuous. Consider a point $y\in [0,1]$ such that both $Q^*$ and $G$ are differentiable at $y$. The set of such points has measure 1 because both $Q^*$ and $G$ are absolutely continuous and almost surely differentiable. Then, the Chain Rule of derivatives applies and we get that almost surely 
	\begin{align*}
		u'(y)\ &=\ - (Q^*)'(y) \cdot G(y)\ +\ (v - Q^*(y)) \cdot G'(y)\\
		&= (v - F^-(y)) \cdot G'(y) \ +\ (F^-(y) - Q^*(y)) \cdot G'(y)\ -\ (Q^*)'(y) \cdot G(y) \,.\\
        &= (v - F^-(y)) \cdot \frac{G(y)}{1 - F(Q^*(y))},
	\end{align*} 
    where the last equality holds almost surely by replacing $G'$ with the expression derived in \Cref{lemma:second-ode}.

    Furthermore, as $u$ is absolutely continuous on $[0,1]$, we can write
	\begin{align*}
		u(y)\ =\ u(0)\ +\ \int_0^y u'(t) \cdot dt\,.
	\end{align*}
	
    As $F^-$ is a non-decreasing function, $Y(v)$ is an interval included in $[0,1]$. Let $y_1 \coloneqq \inf\ Y(v)$ and $y_2 \coloneqq \sup Y(v)$ and fix $y^* \in Y(v)$. To complete the proof, we establish the sufficient condition in \eqref{eq:lemma-bid-strat-inter} by showing that $y^* \in \argmax_{y \in [0,1]}\ u(y)$. 
    
    Note that $F^-(y) < v$ for $y < y_1$, $F^-(y) = v$ for $y \in (y_1, y_2)$, and $F^-(y) > v$ for $y > y_2$ Therefore, $u'(t) > 0$ for $t < y_1$, $u'(t) = 0$ for $t \in (y_1, y_2)$, and $u'(t) < 0$ for $y > y_2$. As a consequence,
	\begin{align*}
		u(y^*)\ -\ u(y)\ =\ \int_{y}^{y^*} u'(t) dt\ \geq 0 &&\forall\ y \leq y^*\,,
	\end{align*}
	and
	\begin{align*}
		u(y)\ -\ u(y^*)\ =\ \int_{y^*}^{y} u'(t) dt\ \leq 0 &&\forall\ y \geq y^*\,.
	\end{align*}
	Thus, we have $y^* \in \argmax_{y \in [0,1]}\ u(y)$ as desired.
\end{proof}

\subsection{Proof of \Cref{thm:main-result}}

\begin{proof}[\textbf{Proof of \Cref{thm:main-result}}]
	\Cref{lemma:ode-existence} establishes the existence of a quantile-based bidding strategy $Q^* \in \mathcal{Q}$ with the required properties. Here, we to show that $(Q^*, H^*)$ is a saddle point of $\rr[F]{\cdot, \cdot }$.
	
	For the first part, consider $H^*$ as defined in \Cref{lemma:second-ode} and an arbitrary bidding strategy $s \in \Scal$. Then, \Cref{lemma:bid-strat-optimality} implies that, for every value $v \in [0,1]$, we have
	\begin{align*}
		Y(v)\ \coloneqq\ \{Q^*(y) \mid y \in[0,1],F^-(y) = v\}\ \subseteq\ \argmax_{b \in [0,1]}\ \E_{h \sim H^*}[(v - b) \cdot \mathbbm{1}(b \geq h)]\,.
	\end{align*}
	Therefore, for every value $v \in [0,1]$, we must have
	\begin{align}\label{eq:main-result-inter-1}
		\E_{h \sim H^*}[(v - Q(y)) \cdot \mathbbm{1}(Q(y) \geq h)]\ \geq\ \E_{b \sim s(v)}\left[ \E_{h \sim H^*}[(v - b) \cdot \mathbbm{1}(b \geq h)] \right] &&\forall\ y \in Y(v)\,.
	\end{align}
	Now, for every $v \in [0,1]$ such that $Y(v) \neq \emptyset$, the definition of $\s[Q^*](v)$ (\Cref{def:quantile-based-bidding-strategy}) implies that, when $b \sim \s[Q^*](v)$, there \emph{always} exists some $y \in Y(v)$ such that $b = Q(y)$. Moreover, as $F^-(t) \sim F$ when $t \sim \unif(0,1)$~\citep{embrechts2013note}, we get
	\begin{align*}
		\Prob_{v \sim F}(Y(v) = \emptyset)\ =\ \Prob_{t \sim \unif(0,1)}(Y(F^-(t)) = \emptyset)\ =\ 0\,.
	\end{align*}
	Hence, taking an expectation over $v\sim F$ in \eqref{eq:main-result-inter-1} yields
	\begin{align*}
		\U[F]{\s[Q^*], H^*} = \mathbb{E}_{(v, h) \sim F \times H^*}  \left[ \mathbb{E}_{b \sim \s[Q^*](v)} \left[ u(b,h; v) \right] \right]\ \geq\ \mathbb{E}_{(v, h) \sim F \times H^*}  \left[ \mathbb{E}_{b \sim s(v)} \left[ u(b,h; v) \right] \right] = \U[F]{s, H^*}\,.
	\end{align*}
	As a direct consequence, we get the first part of the saddle point result:
	\begin{align*}
		\rr[F]{\s[Q^*], H^*}\ =\ \mathcal{O}_{F}(H^*) - \U[F]{\s[Q^*], H^*} \leq \mathcal{O}_{F}(H^*) - \U[F]{\s, H^*} = \rr[F]{s,H^*} \quad \forall\ s\in \Scal\,.
	\end{align*}
	
	For the second part, consider the bidding strategy corresponding to $Q^*$ as described in \Cref{def:quantile-based-bidding-strategy}. The definition of $\rr[F]{s,H}$ implies that $\rr[F]{s,h} = \E_{h \sim H}[\rr[F]{s,h}]$. Since the maximum value of a random variable is larger than its expectation, it suffices to show that
	\begin{align*}
		\rr[F]{\s[Q^*], H^*} \geq \rr[F]{\s[Q^*],h} \quad \forall\ h \in [0,1]\,.
	\end{align*}
	Furthermore, as $\s[Q^*]$ never bids above $Q^*(1)$, we get for all $h > Q^*(1)$ that,
	\begin{align*}
		\rr[F]{\s[Q^*], h} &= \mathcal{O}_{F}(h) - \U[F]{\s[Q^*], h}\\
        &\stackrel{(a)}{=} \mathcal{O}_{F}(h)  \stackrel{(b)}{\leq} \mathcal{O}_{F}(Q^*(1)) \stackrel{(c)}{=} \mathcal{O}_{F}(Q^*(1)) - \U[F]{\s[Q^*], Q^*(1)} = \rr[F]{\s[Q^*], Q^*(1)}\,
	\end{align*}
    where $(a)$ and $(c)$ holds because for every $h' \geq Q^*(1)$, the characterization of the induced bid distribution of $\s[Q^*]$ derived in \Cref{lemma:quantile-based-bidding} implies that $\Prob_{v \sim F, b \sim \s[Q^*]}(b < Q^*(1)) = 1$ which implies that  $\U[F]{\s[Q^*], h'} = 0 $ and $(b)$ holds because $\mathcal{O}_{F}(\cdot)$ is non-increasing. 
    Hence, it suffices to show
	\begin{align*}
		\rr[F]{\s[Q^*], H^*} \geq \rr[F]{\s[Q^*],Q^*(y)} \quad \forall\ y \in [0,1]\,.
	\end{align*}
	
	Now, we obtain by Riemann-Stieltjes integration by part that,
	\begin{align*}
		\mathcal O_{F}(h)\ =\ \E_{v \sim F}[(v - h)\mathbbm{1}(v \geq h)] = \ \int_h^1 (t - h) \cdot dF(t) =\ (1 - h) -  \int_h^1 F(t) dt =\ \int_h^1 (1 - F(t)) dt\,.
	\end{align*}
	By combining this equality with \Cref{lemma:quantile-based-bidding}, we rewrite the regret for all $y \in [0,1]$ as,
	\begin{align*}
		 \rr[F]{\s[Q^*],Q^*(y)}\ =\ \mathcal{O}_{F}(Q^*(y))\ -\ \U[F]{\s[Q^*], Q^*(y)} =\  \int_{Q^*(y)}^1 (1 - F(t))dt -\ \int_y^1 (F^-(t) - Q(t))dt\,.
	\end{align*}
	Finally, \Cref{lemma:saddle-max-prob} implies that all $y \in [0,1]$ satisfy
	\begin{align*}
		\int_{Q^*(y)}^1 (1 - F(t)) dt\ -\ \int_y^1 (F^-(t) - Q^*(t)) dt\ =\ \int_0^1 Q^*(t) dt\,.
	\end{align*}
	As $H^*$ is supported on $[0, Q^*(1)] = \{Q(y) \mid y \in [0,1]\}$, we get
	\begin{align*}
		\rr[F]{\s[Q^*], H^*} = \int_0^1 Q^*(t) dt = \rr[F]{\s[Q^*],Q^*(y)} \quad \forall\ y \in [0,1]\,,
	\end{align*}
	thereby establishing the second part of the saddle-point result.
\end{proof}

\subsection{Proof of \Cref{cor:partial-info-saddle}}

\begin{proof}[\textbf{Proof of \Cref{cor:partial-info-saddle}}]
	Observe that $\reg_F(s, \delta_h) = \rr[F]{s,h}$ for all $h \in [0,1]$, i.e., the partial-information and full-information are identical when the highest competing bid $h$ is deterministic. Therefore, for all $s \in \Scal$, we have
	\begin{align*}
		\reg_F(s, \mathcal H^*) = \E_{h \sim H^*}[\reg_F(s, \delta_h)] = \E_{h \sim H^*}[\rr[F]{s, h}] = \rr[F]{s, H^*}\,.
	\end{align*}
	Hence, \Cref{thm:main-result} immediately implies
	\begin{align}\label{eq:cor-partial-info-inter-1}
		\reg_F(\s[Q^*], \mathcal H^*) = \rr[F]{\s[Q^*], H^*} \leq \inf_{s \in \Scal}\ \rr[F]{s, H^*} = \inf_{s \in \Scal} \reg_F(s, \mathcal H^*) \leq  \inf_{s \in \Scal} \sup_{H \in \Delta([0,1])} \reg_F(s, H)\,.
	\end{align}
	
	Next, note that the bid distribution $P_{\s[Q^*], F}$ induced by $\s[Q^*]$ is absolutely continuous. Indeed, \Cref{lemma:quantile-based-bidding} implies that $P_{\s[Q^*], F} = \{Q^*(t) \mid t \sim \unif(0,1)\}$, and the latter is absolutely continuous because $Q^*$ is strictly increasing and absolutely continuous. Therefore, \Cref{thm:full-info} applies, and we get
	\begin{align}\label{eq:cor-partial-info-inter-2}
		\inf_{s \in \Scal} \sup_{H \in \Delta([0,1])} \reg_F(s, H) \leq \sup_{H}\ \reg_F(\s[Q^*], H) = \sup_{h \in [0,1]} \rr[F]{\s[Q^*], h} \leq \rr[F]{\s[Q^*], H^*} = \reg_F(\s[Q^*], \mathcal H^*)\,.
	\end{align}
    Combining \eqref{eq:cor-partial-info-inter-1} and \eqref{eq:cor-partial-info-inter-2} immediately yields
    \begin{align*}
        \inf_{s \in \Scal} \sup_{H \in \Delta([0,1])} \reg_F(s, H)\ =\ \reg_F(\s[Q^*], \mathcal H^*) = \rr[F]{\s[Q^*], H^*} = \int_0^1 Q^*(t)dt\,. \qquad \qedhere
    \end{align*}
\end{proof}

\section{Proofs of Results in \Cref{sec:value-dist-impact}}

\begin{proof}[\textbf{Proof of \Cref{thm:worst-value-dist}}]
	Define $\mathcal Q_0 \coloneqq \{Q \in \mathcal {Q} \mid Q(0) = 0\}$. In \Cref{appendix:atom-at-0}, we show that for every value distribution $F$ such that $F(0)>0$, there exists another distribution $\tilde F$ with $\tilde F(0) = 0$ such that
    \begin{align*}
        \inf_{s \in \Scal}\ \sup_{H \in \Delta([0,1])}\ \reg_F(s,H) \leq \inf_{s \in \Scal}\ \sup_{H \in \Delta([0,1])}\ \reg_{\tilde F}(s,H)\,.
    \end{align*}
    Thus, without loss of generality, we will assume $F(0) = 0$ for all value distributions in this proof.

	Fix any value distribution $F \in \Delta([0,1])$ with $F(0)= 0$. Using \Cref{thm:full-info} and Corollary~\ref{cor:partial-info-saddle}, and the fact that $\s[Q^*] \in \mathcal{Q}_0$, we get
	\begin{align*}
		\inf_{s \in \Scal}\ \sup_{H \in \Delta([0,1])}\ \reg_F(s,H)\ =\ \inf_{Q \in \mathcal Q_0}\ \sup_{H \in \Delta([0,1])}\ \reg_F(\s[Q],H)\ = \inf_{Q\in \mathcal Q_0}\ \sup_{h \in [0,1]}\ \rr[F]{\s[Q], h}\,.
	\end{align*}
	
	Furthermore, for any $Q \in \mathcal{Q}$, as $\s[Q]$ never bids above $Q(1)$, we have for all $h > Q(1)$ that,
	\begin{align*}
		\rr[F]{\s[Q], h} = \mathcal{O}_{F}(h) - \U[F]{\s[Q], h} \leq \mathcal{O}_{F}(Q(1)) = \mathcal{O}_{F}(h) - \U[F]{\s[Q], h} = \rr[F]{\s[Q], Q(1)}\,.
	\end{align*}
    Here, we have used the fact that $\Prob_{v \sim F, b \sim \s[Q]}(b < Q(1)) = 1$, which follows from the characterization of the induced bid distribution of $\s[Q]$ in \Cref{lemma:quantile-based-bidding}. Therefore,
	\begin{align*}
		\inf_{s \in \Scal}\ \sup_{H \in \Delta([0,1])}\ \reg_F(s,H)\ =\ \inf_{Q\in \mathcal Q_0}\ \sup_{y \in [0,1]}\ \rr[F]{\s[Q], Q(y)}\,.
	\end{align*}
	
	Now, for $Q \in \mathcal{Q}_0$ and $y \in [0,1]$, we have
	\begin{align*}
		\rr[F]{\s[Q], Q(y)}\ &=\ \mathcal{O}_{F}(Q(y))\ -\ \U[F]{\s[Q], Q(y)}\\
		&=\ \E_{v \sim F}[(v - Q(y))\mathbbm{1}(v\geq Q(y))]\ -\ \int_y^1 (F^-(t) - Q(t)) dt \tag{\Cref{lemma:quantile-based-bidding}}\\
		&=\ \E_{t \sim \unif(0,1)}[(F^-(t) - Q(y))\mathbbm{1}(F^-(t)\geq Q(y))]\ -\ \int_y^1 (F^-(t) - Q(t)) dt\\
		&=\ \int_{F(Q(y))}^1 (F^-(t) - Q(y)) dt\ -\ \int_y^1 (F^-(t) - Q(t)) dt\\
		&=\ \int_{F(Q(y))}^y (F^-(t) - Q(y)) dt\ +\ \int_y^1 (Q(t) - Q(y)) dt\,.
	\end{align*}
	
	Therefore,
	\begin{align*}
		\inf_{s \in \Scal}\ \sup_{H \in \Delta([0,1])}\ \reg_F(s,H)\ =\ \inf_{Q\in \mathcal Q_0}\ \sup_{y \in [0,1]}\ \int_{F(Q(y))}^y (F^-(t) - Q(y)) dt\ +\ \int_y^1 (Q(t) - Q(y)) dt\,.
	\end{align*}
	
	As $\s[Q^*]$ is a minimax-optimal strategy and $Q^*(y) \leq F^-(y)$ for all $y \in [0,1]$, we get
	\begin{align*}
		&\inf_{Q\in \mathcal Q_0}\ \sup_{y \in [0,1]}\ \int_{F(Q(y))}^y (F^-(t) - Q(y)) dt\ +\ \int_y^1 (Q(t) - Q(y)) dt\\
		=\ &\inf_{Q\in \mathcal Q_0}\ \sup_{\substack{y \in [0,1]: \\ F(Q(y)) \leq y}}\ \int_{F(Q(y))}^y (F^-(t) - Q(y)) dt\ +\ \int_y^1 (Q(t) - Q(y)) dt\\
		=\ &\inf_{Q\in \mathcal Q_0}\ \sup_{\substack{y \in [0,1]: \\ F(Q(y)) \leq y}}\ \int_{F(Q(y))}^y (F^-(t) - Q(y))^+ dt\ +\ \int_y^1 (Q(t) - Q(y)) dt\\
		\leq\ &\inf_{Q\in \mathcal Q_0}\ \sup_{y \in [0,1]}\ \int_{F(Q(y))}^y (F^-(t) - Q(y))^+ dt\ +\ \int_y^1 (Q(t) - Q(y)) dt\\
	\end{align*}
	
	Hence, we have shown that the following statement holds for all $F$ such that $F(0) = 0$:
	\begin{align*}
		\inf_{s \in \Scal}\ \sup_{H \in \Delta([0,1])}\ \reg_F(s,H)\ \leq\ \inf_{Q\in \mathcal Q_0}\ \sup_{y \in [0,1]}\ \int_{F(Q(y))}^y (F^-(t) - Q(y))^+ dt\ +\ \int_y^1 (Q(t) - Q(y)) dt\,. \tag{\#}
	\end{align*}
	
	Next, fix a $\rho \in [1,\infty]$ and set $F_\rho \coloneqq \unif(1 - \tfrac{1}{\rho}, 1)$, where $F_\infty = \delta_1$. Taking an supremum over $F \in \mathcal F_\rho$ on both sides of $(\#)$ yields
	\begin{align*}
		&\sup_{F \in \F_\rho}\ \inf_{s \in \Scal}\ \sup_{H \in \Delta([0,1])}\ \reg_F(s,H)\\
		\leq\ &\sup_{F \in \F_\rho}\ \inf_{Q\in \mathcal Q_0}\ \sup_{y \in [0,1]}\ \int_{F(Q(y))}^y (F^-(t) - Q(y))^+ dt\ +\ \int_y^1 (Q(t) - Q(y)) dt\\
		\leq\ &\inf_{Q\in \mathcal Q_0}\ \sup_{F \in \F_\rho}\  \sup_{y \in [0,1]}\ \int_{F(Q(y))}^y (F^-(t) - Q(y))^+ dt\ +\ \int_y^1 (Q(t) - Q(y)) dt\\
		=\ &\inf_{Q\in \mathcal Q_0}\ \sup_{y \in [0,1]}\ \sup_{F \in \F_\rho}\   \int_{F(Q(y))}^y (F^-(t) - Q(y))^+ dt\ +\ \int_y^1 (Q(t) - Q(y)) dt
	\end{align*}
	
	If we can show that
	\begin{align*}
		F_\rho\ \in \argmax_{F \in \F_\rho} \int_{F(Q(y))}^y (F^-(t) - Q(y))^+ dt &&\forall\ Q\in \mathcal{Q}_0,\ y\in [0,1]\,, \tag{$\diamond$}
	\end{align*}
	then we would be done because we would have shown
	\begin{align*}
		\sup_{F \in \F_\rho}\ \inf_{s \in \Scal}\ \sup_{H \in \Delta([0,1])} \reg_F(s,H)\ &\leq\ \inf_{Q\in \mathcal Q_0}\sup_{y \in [0,1]}\   \int_{F(Q(y))}^y (F^-(t) - Q(y))^+ dt\ +\ \int_y^1 (Q(t) - Q(y)) dt\\
		&=\ \inf_{s \in \Scal}\ \sup_{H \in \Delta([0,1])} \reg_{F_{\rho}}(s,H)\,.
	\end{align*}
	
	To finish the proof, we now establish $(\diamond)$. We do so in two steps: we fix $F \in \mathcal{F}_{\rho}$ and (i) we show that $F_\rho(h) \leq F(h)$ for all $h \in [0,1]$; (ii)  we show that $F^-(t) \leq F_\rho^-(t)$ for all $t \in [0,1]$.
	\begin{itemize}
		\item[(i)] Fix any $h \in [0,1]$. Then, $$F(h) = 1 - F((h,1]) \geq 1 - \min\{1, \rho\cdot \lambda((h,1])\} = 1 - F_\rho((h, 1]) = F_\rho(h)\,.$$
		\item[(ii)] As $F^-(0) = F^-_\rho(0)$, we focus on $t \in (0,1]$. Then, $F_\rho^-(t) = 1 + (t-1)/\rho$ and 
		\begin{align*}
			F(F^-_\rho(t)) = F\left(1 + \frac{t-1}{\rho} \right) = 1 - F\left(\left(1 + \frac{t-1}{\rho}, 1\right] \right) \geq 1 - \rho \cdot \lambda\left(\left(1 + \frac{t-1}{\rho}, 1\right] \right) = 1- (t-1) = t\,.
		\end{align*}
		As a consequence, we get $F^-(t) \leq F_\rho^-(t)$.
	\end{itemize}
	
	Finally, (i) and (ii) together imply $(\diamond)$ because
	\begin{align*}
		\int_{F(Q(y))}^y (F^-(t) - Q(y))^+ dt \leq \int_{F_\rho(Q(y))}^y (F_\rho^-(t) - Q(y))^+ dt &&\forall\ Q\in \mathcal{Q}_0,\ y\in [0,1]\,.
	\end{align*}  
	As proving $(\diamond)$ was sufficient to complete the proof, we have established the theorem.
\end{proof}
\section{Alternative Tie-Breaking Rules}\label{appendix:tie-breaking}

In our model (\Cref{sec:model}), we assumed that the utility is given by
\begin{align*}
    u(b,h;v) = (v - b) \mathbbm{1}(b \geq h)\,.
\end{align*}

This assumed that ties are broken in favor of the buyer under consideration. Here, we show that our results continue to hold for all other tie-breaking rules. Consider an arbitrary tie-breaking rule which yields an expected utility of $\pi(v,h)$ whenever the buyer bids $h$ and ties with the highest competing bid. Here, $\pi(v,h) \leq (v-h)$ because $v-h$ is the utility obtained when the tie is broken completely in favor of the buyer. The utility under this tie-breaking rule is given by
\begin{align*}
    u^\pi(b,h;v) = (v - b) \mathbbm{1}(b > h) + \pi(v,b) \mathbbm{1}(b = h)\,,
\end{align*}
and it satisfies $u^\pi(b,h;v) \leq u(b,h;v)$ for all $b,h,v \in [0,1]$.

We start by showing that the tie-breaking rules does not affect the optimal utility which can be attained against any highest-competing-bid distribution.

\begin{lemma}\label{lemma:tie-breaking-bechmark}
    For every highest-competing-bid distribution $H \in \Delta([0,1])$, we have
    \begin{align*}
        \sup_{s' \in \Scal}\ \E_{(v,h) \sim F \times H}\left[ \E_{b \sim s'(v)}[u^\pi(b,h;v)] \right] = \sup_{s' \in \Scal}\ \E_{(v,h) \sim F \times H}\left[ \E_{b \sim s'(v)}[u(b,h;v)] \right] = \sup_{s' \in \Scal} \U[F]{s', H}
    \end{align*}
\end{lemma}
\begin{proof}
    Since $u^\pi(b,h;v) \leq u(b,h;v)$, we have
    \begin{align*}
        \sup_{s' \in \Scal}\ \E_{(v,h) \sim F \times H}\left[ \E_{b \sim s'(v)}[u^\pi(b,h;v)] \right] \leq \sup_{s' \in \Scal}\ \E_{(v,h) \sim F \times H}\left[ \E_{b \sim s'(v)}[u(b,h;v)] \right]\,.
    \end{align*}
    For contradiction, suppose the inequality is strict. Then, there exists $\epsilon > 0$ and a strategy $s \in \Scal$ such that
    \begin{align*}
        \E_{(v,h) \sim F \times H}\left[ \E_{b \sim s'(v)}[u^\pi(b,h;v)] \right] < \E_{(v,h) \sim F \times H}\left[ \E_{b \sim s(v)}[u(b,h;v)] \right]\ -\ \epsilon \qquad \forall\ s' \in \Scal\, \tag{\#}
    \end{align*}
    As bidding higher that the value is never beneficial, we can assume without loss of generality that we always have $b \leq v$ when $b \sim s(v)$. 
    
     Now, define the strategy $s'$ to be the one which always bids $\epsilon$ more than the bidding strategy $s$ whenever possible, i.e., $s'(v)$ is the distribution of $\min\{b+\epsilon, 1\}$ when $b \sim s(v)$.  Observe that, for all $b,h,v \in [0,1]$ such that $b \leq v$, we have 
    \begin{align*}
        u^\pi(b + \epsilon,h;v) \geq (v - b - \epsilon) \mathbbm{1}(\min\{b + \epsilon, 1\} > h) \geq (v - b) \mathbbm{1}(b \geq h)\ -\ \epsilon\,.
    \end{align*}
    Therefore, we get
    \begin{align*}
        \E_{(v,h) \sim F \times H}\left[ \E_{b \sim s'(v)}[u^\pi(b,h;v)] \right] &= \E_{(v,h) \sim F \times H}\left[ \E_{b \sim s(v)}[u^\pi(b+\epsilon,h;v)] \right]\\
        &\geq \E_{(v,h) \sim F \times H}\left[ \E_{b \sim s(v)}[u(b,h;v)] \right]\ -\ \epsilon\,.
    \end{align*}
    This yields the desired contradiction to $(\#)$.
\end{proof}

Next, we show that the tie-breaking rule does not impact utility under any bidding strategy which induces a continuous bid distribution.

\begin{lemma}\label{lemma:tie-breaking-cont-strat}
    Every bidding strategy $s \in \Scal$ which induces an absolutely continuous bid distribution $P_{s,F}$ satisfies
    \begin{align*}
        \E_{(v,h) \sim F \times H}\left[ \E_{b \sim s(v)}[u^\pi(b,h;v)] \right] = \E_{(v,h) \sim F \times H}\left[ \E_{b \sim s(v)}[u(b,h;v)] \right] = \U[F]{s,F} \qquad \forall\ H \in \Delta([0,1])\,.
    \end{align*}
\end{lemma}
\begin{proof}
    We show that ties are a zero probability event in these conditions. Note that, for any highest competing bid $h$, the absolute continuity of $P_{s,F}$ implies
    \begin{align*}
        \Prob_{v \sim F, b \sim s(v)}(b = h) = \Prob_{b \sim P_{s,F}}(b = h) = 0\,.
    \end{align*}
    Therefore, for any fixed $h$, we have $u^\pi(b,h;v) = u(b,h;v)$ almost surely for $v \sim F, b \sim s(v)$. The lemma follows as a consequence.
\end{proof}

Combining the two lemmas allows us to establish the minimax optimality of our strategy $\s[Q^*]$ for all tie-breaking rules. Before stating the result, we introduce one new piece of notation: define the worst-case regret under the alternative tie-breaking rule $\pi$ as follows
\begin{align*}
    \mathrm{WReg}^\pi_F(s) \coloneqq \sup_{H \in \Delta([0,1]}\ \sup_{s' \in \Scal}\ \E_{(v,h) \sim F \times H}\left[ \E_{b \sim s'(v)}[u^\pi(b,h;v)] \right] - \E_{(v,h) \sim F \times H}\left[ \E_{b \sim s(v)}[u^\pi(b,h;v)] \right]\,.
\end{align*}
Now, \Cref{lemma:tie-breaking-bechmark} and $u^\pi(b,h;v) \leq u(b,h;v)$ imply
\begin{align}\label{eq:tie-breaking-wreg-ineq}
    \mathrm{WReg}^\pi_F(s) \geq \mathrm{WReg}_F(s) \qquad \forall\ s \in \Scal\,.
\end{align}
The following proposition shows that they attain the same optimal worst-case regret, and do so with our minimax-optimal strategy.

\begin{proposition}\label{prop:tie-breaking-wreg}
    For every value distribution $F$ and tie-breaking rule $\pi$, we have
    \begin{align*}
        \mathrm{WReg}^\pi_F(\s[Q^*])\ =\ \inf_{s \in \Scal}\  \mathrm{WReg}^\pi_F(s)\ =\  \inf_{s \in \Scal}\  \mathrm{WReg}_F(s)\ =\ \mathrm{WReg}_F(\s[Q^*])\,,
    \end{align*}
    where $\s[Q^*]$ is the minimax-optimal bidding strategy described in \Cref{thm:main-result}.
\end{proposition}
\begin{proof}
    In light of \eqref{eq:tie-breaking-wreg-ineq} and \Cref{cor:partial-info-saddle}, it suffices to show that
    \begin{align*}
        \mathrm{WReg}^\pi_F(\s[Q^*])\ =\ \mathrm{WReg}_F(\s[Q^*])\,.
    \end{align*}
    Recall that \Cref{lemma:quantile-based-bidding} implies the absolute continuity of bidding strategies based on quantile-based bidding strategies. $\s[Q^*]$ is one such strategy, and therefore it induces an absolutely continuous bid distributions. Therefore, \Cref{lemma:tie-breaking-cont-strat} applies and we get
    \begin{align*}
        \mathrm{WReg}^\pi_F(\s[Q^*])\ &=\ \sup_{H \in \Delta([0,1]}\ \sup_{s' \in \Scal}\ \E_{(v,h) \sim F \times H}\left[ \E_{b \sim s'(v)}[u^\pi(b,h;v)] \right] - \E_{(v,h) \sim F \times H}\left[ \E_{b \sim \s[Q^*](v)}[u^\pi(b,h;v)] \right]\\
        &=\ \sup_{H \in \Delta([0,1]}\  \sup_{s' \in \Scal} \U[F]{s', H} - \U[F]{\s[Q^*], H}\\
        &=\ \mathrm{WReg}_F(\s[Q^*]) \qedhere
    \end{align*}
\end{proof}

Thus, we have shown that $\s[Q^*]$ is a minimax-optimal bidding strategy for all tie-breaking rules $\pi$, as desired.

\section{$\mathbf{F(0) = 0}$ is Without Loss of Generality}\label{appendix:atom-at-0}

We  assumed that the value distribution $F$ satisfied $F(0) = 0$ in \Cref{sec:minimax}. Here, we establish our claim that this assumption is without loss of generality. In particular, we show that \Cref{thm:main-result} can be also used to characterize the minimax-optimal strategy, and its regret, even for value distributions which do not satisfy $F(0) = 0$. Note that, in the scenario where $v = 0$ with probability 1, always bidding zero is minimax optimal and the minimax regret is zero. Therefore, we only consider value distributions which produce a positive value with positive probability.

\begin{proposition}\label{prop:atom-at-0}
	Consider an arbitrary value distribution $\tilde F \in \Delta([0,1])$ with $\tilde F(0) = a < 1$, and let $F$ be the distribution of the value $v \sim F$ conditioned on it being non-zero, i.e., $F(t) = (\tilde F(t) - a)/(1-a)$ for all $t \in [0,1]$. Let $\s[Q^*]$ be the minimax strategy corresponding to the value distribution $F$, as given in Theorem~\ref{thm:main-result}. Then, the bidding strategy $\tilde \s_{Q^*}$ defined as
	\begin{align*}
		\tilde s_{Q^*}(v) = \begin{cases}
			\delta_0 &\text{if } v = 0\\
			\s[Q^*](v) &\text{if } v \in (0,1]\\
		\end{cases}
	\end{align*}
	is a minimax-optimal bidding strategy and satisfies
	\begin{align*}
		\sup_{H \in \Delta([0,1])} \reg_{\tilde F}(\tilde s_{Q^*}, H) = \inf_{s \in \Scal} \sup_{H \in \Delta([0,1])} \reg_{\tilde F}(s, H) = (1 - a) \cdot \sup_{H \in \Delta([0,1])} \reg_F(\s[Q^*], H)\,.
	\end{align*}
\end{proposition}

\begin{proof}
	
	Let $\Scal_0 \subset \Scal$ be the set of strategies which always bid $0$ when the value is $0$.
	
	Observe that $u(0,h;0) = 0$ for all $h \in [0,1]$. Therefore, for any strategy $s \in \Scal_0$ and any highest-competing-bid distribution $H \in \Delta([0,1])$, we have
	\begin{align*}
		\U[\tilde F]{s, H} &= \mathbb{E}_{(v, h) \sim \tilde F \times H}  \left[ \mathbb{E}_{b \sim s(v)} \left[ u(b,h;v) \right] \right]\\
		&= \mathbb{E}_{(v, h) \sim \tilde F \times H}  \left[ \mathbb{E}_{b \sim s(v)} \left[ u(b,h;v) \right] \cdot \mathbbm{1}(v > 0) \right]\\
		&= \mathbb{E}_{h \sim H}\left[ \E_{v \sim \tilde F}  [ \mathbb{E}_{b \sim s(v)} \left[ u(b,h;v) \right] \mid v > 0] \cdot \Prob(v > 0) \right]\\
		&= (1 - a) \cdot \mathbb{E}_{h \sim H}\left[ \E_{v \sim \tilde F}  [ \mathbb{E}_{b \sim s(v)} \left[ u(b,h;v) \right] \mid v > 0]] \right]\\
		&= (1 - a) \cdot \mathbb{E}_{h \sim H}\left[ \E_{v \sim F} [ \mathbb{E}_{b \sim s(v)} \left[ u(b,h;v) \right] \right]\\
		&= (1 - a) \cdot \U[F]{s, H}\,.
	\end{align*}
	
	Since 0 is the optimal bid for value 0 regardless of the highest-competing-bid distribution, we must have
	\begin{align*}
		\sup_{s' \in \Scal} \U[\tilde F]{s',H} = \sup_{s' \in \Scal_0} \U[\tilde F]{s', H} = (1 - a) \cdot \sup_{s' \in \Scal_0} \U[F]{s', H} = (1 - a) \cdot \sup_{s' \in \Scal} \U[F]{s',H}\,.
	\end{align*}
	
	Therefore, for every $H$ and $s \in \Scal_0$, we have
	\begin{align*}
		\reg_{\tilde F}(s, H) = \sup_{s' \in \Scal} \U[\tilde F]{s',H} - \U[\tilde F]{s,H} = (1 -a) \cdot \left\{ \sup_{s' \in \Scal} \U[F]{s',H} - \U[F]{s,H} \right\} = (1 -a)\cdot \reg_F(s, H).
	\end{align*}
	Taking a supremum over $H$ yields
	\begin{align*}
		\sup_{H \in \Delta([0,1])} \reg_{\tilde F}(s, H) = (1 -a)\cdot \sup_{H \in \Delta([0,1])} \reg_F(s, H) \qquad \forall s \in \Scal_0\,.
	\end{align*}
	
	Now, for any strategy $s \in \Scal$, if we define $\tilde s \in \Scal_0$ as $\tilde s(0) = 0$ with probability 1 and $\tilde s(v) = s(v)$ for all $v > 0$, then $\U[\tilde F]{s, H} \leq \U[\tilde F]{\tilde s, H}$ because zero is the optimal bid for the value $v = 0$ against all $h \sim H$, and any other bid can only do worse. Therefore, we have
	\begin{align*}
		\inf_{s \in \Scal} \sup_{H \in \Delta([0,1])} \reg_{\tilde F}(s, H) = \inf_{s \in \Scal_0} \sup_{H \in \Delta([0,1])} \reg_{\tilde F}(s, H) = (1 - a) \cdot \inf_{s \in \Scal_0} \sup_{H \in \Delta([0,1])} \reg_{F}(s, H)\,.
	\end{align*}
	\Cref{cor:partial-info-saddle} implies that $\s[Q^*]$ is an optimal strategy for the minimax problem given in the last term. As a consequence, we have
	\begin{align*}
		\inf_{s \in \Scal} \sup_{H \in \Delta([0,1])} \reg_{\tilde F}(s, H) = (1 - a) \cdot \sup_{H \in \Delta([0,1])} \reg_{F}(\s[Q^*], H)\,.
	\end{align*}
	Finally, the definitions of $\s[Q^*]$ and $\tilde s_{Q^*}$ imply $\U[F]{\tilde s_{Q^*}, H} = \U[F]{s_{Q^*}, H}$ for all $H \in \Delta([0,1])$. Hence, we get
	\begin{align*}
		\sup_{H \in \Delta([0,1])} \reg_{\tilde F}(\tilde s_{Q^*}, H) = (1 -a) \cdot \sup_{H \in \Delta([0,1])} \reg_{F}(\tilde s_{Q^*}, H)
		= (1 -a) \cdot \sup_{H \in \Delta([0,1])} \reg_{F}(s_{Q^*}, H)\,,
	\end{align*}
	which completes the proof.
\end{proof}

\section{Uniform-bid-shading}\label{sec:uniform-bid-shading}

In this section, we provide a stronger characterization of the worst-case regret for uniform-bid-shading strategies under additional assumptions on the value distribution. Namely, we prove the following result.
\begin{proposition}\label{prop:uniform-bid-shading}
	Consider a value distribution $F$ with a density $f:[0,1] \to [0,1]$ such that the map $t \mapsto t\cdot f(t)$ is non-decreasing. Then, for any shading factor $\alpha \in [0,1]$,
	\begin{align*}
		\sup_{H \in \Delta([0,1])}\ \reg_F(s_\alpha, H)\ =\ \max\left\{\alpha \cdot \E_{v \sim F}[v],\ \E_{v\sim F}\left[(v - \alpha \cdot F^-(1))^+ \right]\right\}\,.
	\end{align*}
\end{proposition}
We note that the assumption on the value distribution holds for any distribution $\unif(1 - \tfrac{1}{\rho}, 1)$, with $\rho >0$. This family of value distributions is particularly significant, as it arises as the worst-case scenario when evaluating the impact of informational asymmetries in bidding environments (see \Cref{sec:value-dist-impact}).

\begin{proof}[\textbf{Proof of \Cref{prop:uniform-bid-shading}}]
	As $F$ has a density, it is absolutely continuous, and as a consequence, so is the distribution of bids $\alpha \cdot v$ under $s_\alpha$. Therefore, \Cref{thm:evaluation} applies, and we get
	\begin{align*}
		\sup_{H \in \Delta([0,1])}\ \reg_F(s_\alpha, H)\ =\ \E_{v \sim F}[(v - h) \cdot \mathbbm{1}(v \geq h)]\ - \U[F]{s_\alpha, \delta_h} \,.
	\end{align*}
	
	Consider the mapping $r : h 
     \mapsto \E_{v \sim F}[(v - h) \cdot \mathbbm{1}(v \geq h)]\ - \U[F]{s_\alpha, \delta_h}$. For every $h \in [0, \alpha \cdot F^{-}(1)]$, we have
	\begin{align*}
	 r(h) \ &=\ \E_{v \sim F}[(v - h) \mathbbm{1}(v \geq h)]\ -  \E_{v \sim F}[(v - \alpha \cdot v) \mathbbm{1}(\alpha \cdot v \geq h)]\\
		&=\ \int_h^1 (v - h) \cdot f(v) dv\ -\ \int_{h/\alpha}^1 (1 - \alpha) \cdot v \cdot f(v)  dv\\
		&=\ (1 - h)F(1) - (h-h)F(h)\ - \int_h^1 F(v)dv\ -\ \int_{h/\alpha}^1 (1 - \alpha) \cdot v \cdot f(v)  dv\\
		&=\ \int_h^1(1 - F(v))dv\ - \int_{h/\alpha}^1 (1 - \alpha) \cdot v \cdot f(v)  dv\,.
	\end{align*}
	Therefore, we get
	\begin{align*}
		r'(h)\ =\ (1 - \alpha) \cdot \frac{h}{\alpha} \cdot f\left( \frac{h}{\alpha}\right)\ -\ (1 - F(h))\,.
	\end{align*}
	As $t \mapsto t \cdot f(t)$ is assumed to be non-decreasing, and $t \mapsto 1 - F(t)$ is non-increasing for all distributions, we get that $r'(\cdot)$ is non-decreasing for $h \in [0, \alpha \cdot F^-(1)]$. Or equivalently, $r(\cdot)$ is a convex function for $h \in [0, \alpha \cdot F^-(1)]$. Moreover, $r(\cdot)$ is also continuous. Therefore, Bauer Maximum Principle (see 7.69 of \citealt{aliprantis2006infinite}) applies and we get
	\begin{align*}
		\sup_{h \in [0, \alpha \cdot F^-(1)]}\ r(h) \ =\ \max\{r(0), r(\alpha \cdot F^-(1))\}\,.
	\end{align*}
	Since $s_\alpha(v) \in [0,\alpha \cdot F^-(1)]$ for all values $v \in [0,1]$, we have for all $h > \alpha \cdot F^-(1)$ that,
	\begin{align*}
		r(h) \ &= \ \E_{v \sim F}[(v - h) \cdot \mathbbm{1}(v \geq h)]\ - \U[F]{s_\alpha, \delta_h}\\
        &=\ \E_{v \sim F}[(v - h) \cdot \mathbbm{1}(v \geq h)]\\
        &\geq\ \E_{v \sim F}[(v - \alpha \cdot F^{-}(1)) \cdot \mathbbm{1}(v \geq \alpha \cdot F^{-}(1))]\\
        &= \ \E_{v \sim F}[(v - \alpha \cdot F^{-}(1)) \cdot \mathbbm{1}(v \geq \alpha \cdot F^{-}(1))] \ - \ \U[F]{s_\alpha, \alpha \cdot F^-(1)} \\ 
        &=\ r \left( \alpha \cdot F^{-}(1) \right) \,.
	\end{align*}
	Therefore, $r(h) \leq r(\alpha \cdot F^-(1))$ for all $h > \alpha \cdot F^-(1)$. Altogether, we get
	\begin{align*}
		\sup_{H \in \Delta([0,1])}\ \reg_F(s_\alpha, H)\ =\ \sup_{h \in [0,\alpha \cdot F^-(1)]}\ r(h) \ =\ \max\{  r(0), r(\alpha \cdot F^-(1))\}\,.
	\end{align*}
	
	Finally, observe that
	\begin{align*}
		r(0)\ =\ \E_{v \sim F}[v] - \E_{v \sim F}[v - \alpha \cdot v]\ =\ \alpha \cdot \E_{v \sim F}[v]\,,
	\end{align*}
	and
	\begin{align*}
		r( \alpha \cdot F^-(1))\ =\ \E_{v \sim F}[(v - \alpha \cdot F^-(1)) \mathbbm{1}(v \geq \alpha\cdot F^-(1))] - 0\ =\ \E_{v\sim F}\left[(v - \alpha \cdot F^-(1))^+ \right]\,.\ &\qedhere
	\end{align*}
\end{proof}

\end{document}